\documentclass[journal,comsoc]{IEEEtran}
\usepackage{graphicx}
\usepackage{amssymb}
\usepackage{amsmath}
\usepackage{algcompatible}
\usepackage{algorithm}
\usepackage{algpseudocode}
\usepackage{algpseudocodex}
\usepackage{array}
\usepackage{enumitem}
\usepackage{textcomp}
\usepackage{physics}
\usepackage{stfloats}
\usepackage{url}
\usepackage{hyperref}
\usepackage{verbatim}
\usepackage{booktabs}
\usepackage{tabularx}
\usepackage{cite}
\usepackage[short, c3]{optidef}
\usepackage[mathlines]{lineno}
\usepackage{color}
\usepackage{soul}
\usepackage{threeparttable}
\usepackage{tabulary}
\usepackage{multirow}
\usepackage{pbox}
\usepackage{multicol}
\usepackage{lipsum}
\usepackage{amsthm}
\usepackage{relsize}
\usepackage{epstopdf}
\usepackage{mathtools}
\usepackage{calc}
\usepackage{makecell}
\usepackage{pgfplots}\pgfplotsset{compat=1.18}
\usepackage[caption=false]{subfig}
\newtheorem{definition}{Definition}
\newtheorem{theorem}{Theorem}
\newtheorem{lemma}{Lemma}
\newtheorem{corollary}{Corollary}

\newtheorem{remark}{Remark}

\newcommand{\myrule}{\vspace{-0.6\baselineskip}\hrulefill\vspace{-0.1\baselineskip}}

\begin{document}

\title{\huge Joint Load Balancing and {Transmit} Power Control for Energy Efficiency Maximization {in the} Satellite-Cell-Free Massive MIMO Uplink}
\author{Ngo Tran Anh Thu,  \textit{Student Member}, \textit{IEEE}, Lo Hai Long, Le Duc Anh Vu, \\ Trinh Van Chien, \textit{Member}, \textit{IEEE} and Lajos Hanzo, \textit{Life Fellow}, \textit{IEEE}
\thanks{Ngo Tran Anh Thu, Lo Hai Long, Le Duc Anh Vu, and Trinh Van Chien are with the School of Information and Communications Technology (SoICT), Hanoi University of Science and Technology (HUST), Vietnam (email: anhthungo.tr@gmail.com, 
hailong2004ptcnn@gmail.com, anhvuleduc@gmail.com, chientv@soict.hust.edu.vn). Lajos Hanzo is with the Department of Electronics and Computer Science, University of Southampton,
U.K. (email: hanzo@soton.ac.uk). This research is funded by the Vietnam Ministry of Education and Training under
project number B2025-BKA-04 for Trinh Van Chien. The financial support of the following Engineering and Physical Sciences Research Council (EPSRC) projects is gratefully acknowledged for Lajos Hanzo: Platform for Driving Ultimate Connectivity (TITAN) (EP/X04047X/1; EP/Y037243/1); Robust and Reliable Quantum Computing (RoaRQ, EP/W032635/1); PerCom  EP/X01228X/1; EP/Y026721/1, India-UK Intelligent Spectrum Innovation ICON UKRI-1859.}
\thanks{\textit{Corresponding author}: Trinh Van Chien.
}}

\markboth{Journal of \LaTeX\ Class Files,~Vol.~18, No.~9, September~2020}%
{How to Use the IEEEtran \LaTeX \ Templates}

\maketitle
\begin{abstract}

The seamless integration of non-terrestrial and terrestrial infrastructures is a key enabler for ubiquitous connectivity in {next-generation (NG)} wireless networks. {We investigate} a hybrid satellite-cell-free Massive MIMO system, where multiple low-Earth-orbit (LEO) satellites jointly serve users {in unison} with terrestrial access points (APs) under {realistic} imperfect channel state information and practical {user} association constraints. We first derive closed-form expressions of the uplink ergodic throughput by exploiting maximum ratio combining (MRC) {for transmission} over spatially correlated Rician fading channels. Our analysis reveals the {characteristic} impact of {both} user-satellite and user-AP association patterns on {both the} spectral efficiency and {rate-}fairness {achieved}. We then formulate an {energy efficiency optimization} problem {under} joint user association and power control. Since the problems are inherently NP-hard due to {the} binary {nature of the user-}association variables, we develop an improved Differential Evolution (IDE) framework that efficiently explores {the} feasible solutions in polynomial time. Numerical results validate our analysis and show that the proposed hybrid scheme substantially improves energy efficiency and network throughput. For large-scale scenarios, the DE framework provides practical user-satellite-AP association guidelines, enabling scalable performance gains.
\end{abstract}


\section{Introduction}\label{sec:Intro}

\textcolor{black}{The evolution towards next-generation (NG) networks is expected to redefine wireless connectivity through the integration of diverse advanced technologies \cite{10054381,7378276,7244171,chen20235g,10640072,10841966}. Specifically, Wang et al.~\cite{10054381} identify terahertz communications, AI-native architectures, integrated sensing and communication (ISAC), and satellite-terrestrial convergence as key NG enablers for achieving ubiquitous three-dimensional coverage. Liu et al.~\cite{7378276} demonstrate that user association in heterogeneous networks requires joint optimization of load balancing and interference management, while Yang and Hanzo~\cite{7244171} show that linear combining schemes achieve near-optimal performance in massive MIMO through favorable propagation and channel hardening. Chen et al.~\cite{chen20235g} further emphasize that the progression from 5G-Advanced to NG demands orders-of-magnitude improvements in spectral efficiency, energy efficiency, and latency through space-ground integration. Furthermore, recent studies have shown that Cell-Free Massive MIMO assisted with reconfigurable intelligent surfaces (RIS) can effectively improve both spectral and energy efficiency. In \cite{10640072}, reflection pattern modulation was integrated with RIS-assisted Cell-Free Massive MIMO, enabling additional data transmission via the data-dependent active RIS block indices and extending RISs from passive reflection to information-bearing transmission, with closed-form spectral efficiency and energy efficiency analysis, along with phase-shift optimization for green uplink communications. STAR-RIS-aided Cell-Free Massive MIMO~\cite{10841966} is studied under realistic hardware impairments (e.g., transceiver distortions, phase noise, imperfect CSI, and spatially correlated fading), with a closed-form spectral efficiency expression and joint beamforming-power control design highlighting its potential for full-space NG coverage.} ISAC \cite{yang2024coordinated} enables the same radio signals to support both data transmission and environmental sensing, facilitating high-precision localization and imaging \cite{zhang2024target}. Near-field communications \cite{liu2025near}, operating in the sub-terahertz and terahertz bands, exploit spherical {wavefront propagation} to achieve ultra-high data rates and spatially focused links. Meanwhile, artificial intelligence (AI) is envisioned to optimize network operation, resource allocation, and maintenance in real time \cite{zuo2023survey}. These innovations collectively support emerging applications such as holographic telepresence, autonomous mobility, and digital twins. Within this ecosystem, Cell-Free Massive MIMO \cite{zheng2024mobile} and {low-Earth-orbit (LEO)} satellite communications \cite{al2022survey} are key enablers of ubiquitous, high-reliability coverage. By leveraging cooperative processing {relying on} distributed APs and satellite gateways, this hybrid space-terrestrial architecture \cite{zhu2021integrated} can deliver high spectral efficiency, robust connectivity, and uniform quality of service across diverse environments.

In Cell-Free Massive MIMO, a large number of distributed access points (APs) {jointly} serve {a set of} users on the same time-frequency resources under centralized coordination, rather than dividing the coverage area into disjoint cells \cite{elhoushy2021cell,ammar2021user,ngo2024ultradense}. This architecture eliminates cell boundaries, mitigates inter-cell interference, and ensures uniform quality of service for users across different locations. By exploiting favorable propagation and channel hardening, the effects of small-scale fading and interference diminish as the number of APs increases \cite{song2021joint}. Both uplink {receiver combining (RC)} and downlink {transmit precoding (TPC)} can be efficiently realized through linear processing schemes such as maximum ratio combining (MRC) and minimum mean-square error (MMSE) detection, enabling scalable implementation {at compellingly} low computational complexity. Moreover, long-term resource allocation, such as power control and user association, can be optimized using large-scale fading statistics to achieve high spectral and energy efficiency.

\begin{table*}[t]
\label{tab:comparedwork}
\caption{Comparative summary between this work and representative studies on massive MIMO systems.}
\resizebox{\linewidth}{!}{%
\begin{tabular}{lccccccccccc|c} 
\hline\cline{1-13}
\textbf{Ref.}                                       & \cite{bashar2020uplink} & \cite{9104036} & \cite{9136914} & \cite{10133823} & \cite{hao2024joint} & \cite{li2024stacked} & \cite{10966047} & \cite{11168825} & \cite{chien2025differential} & \cite{zhu2025holographic} & \cite{li2025holographic} & \textbf{This Work}  \\ 
\hline\hline
\textbf{Year}                                       & 2020                    & 2020           & 2020           & 2023            & 2024                & 2024                 & 2025            & 2025            & 2025                           & 2025                      & 2025                     & \textbf{2026}       \\ 
\hline
\textbf{Massive MIMO}                               & \checkmark              & \checkmark     & \checkmark     & \checkmark      & \checkmark          & \checkmark           &                 & \checkmark      & \checkmark                     & \checkmark                          & \checkmark               & \checkmark          \\
\textbf{Cell-Free}                                  & \checkmark              &                & \checkmark     & \checkmark      & \checkmark          & \checkmark           & \checkmark      &                 & \checkmark                     &                           &                          & \checkmark          \\
\textbf{Satellite}                                  &                         &                &                &                 &                     &                      & \checkmark      &                 &                                &                           & \checkmark               & \checkmark          \\
\textbf{EE Metric}                                  & \checkmark              &                & \checkmark     &                 &                     &                      &                 & \checkmark      &                                &                           &                          & \checkmark          \\
\textbf{Closed-form Rate}                           & \checkmark              & \checkmark     & \checkmark     &                 &                     &                      &                 &                 &                                & \checkmark                & \checkmark               & \checkmark          \\
\textbf{User Association}                           &                         &                &                &                 & \checkmark          &                      & \checkmark      &                 & \checkmark                     &                           & \checkmark               & \checkmark          \\
\textbf{Power Control}                              & \checkmark              & \checkmark     & \checkmark     & \checkmark      & \checkmark          &            &                 & \checkmark      &                      & \checkmark                &                          & \checkmark          \\
\textbf{Evolutionary Computing}                     &                         &                &                & \checkmark      &                     &                      &                 &                 & \checkmark                     &                           &                          & \checkmark          \\
\textbf{Multiple Satellites in Hybrid Architecture} &                         &                &                &                 &                     &                      &                 &                 &                                &                           &                          & \checkmark          \\
\textbf{Joint Discrete-Continuous Opt.}             &                         &                &                &                 &                     &                      &                 &                 &                                &                           &                          & \checkmark          \\
\hline\cline{1-13}
\end{tabular}
}
\end{table*}

Terrestrial Cell-Free Massive MIMO delivers uniform coverage in dense environments but is economically unsustainable for rural or remote areas due to its dense AP requirement~\cite{interdonato2019ubiquitous,ngo2017cell,elhoushy2021cell,zheng2024mobile,Chien2021TWC}. Achieving ubiquitous coverage in NG networks therefore necessitates satellite integration~\cite{luo2024leo}. Geostationary satellites, however, are ill-suited for this role owing to their high latency, cost, and limited spectral reuse~\cite{pan2022latency}. LEO constellations have consequently emerged as the preferred solution, offering low propagation delay and flexible dynamic coverage. In rural areas, LEO satellites bridge coverage gaps where terrestrial rollout entails prohibitive cost; in dense urban scenarios, they relieve spectral congestion and sustain service continuity. This has motivated recent research into hybrid LEO–cell-free Massive MIMO architectures~\cite{van2022space}, which leverage cooperative transmission, centralized processing, and spatial diversity to deliver reliable, fairness-aware ubiquitous service.
Such integration, however, introduces interconnected physical and architectural challenges absent from purely terrestrial deployments. The excessive space-ground path loss of LEO altitudes creates a link budget asymmetry between satellite and terrestrial APs, necessitating high-gain receive arrays and coherent joint processing at the CPU to attain meaningful spatial diversity gains. This asymmetry is further compounded by the satellite-user channel's strong spatial correlation and dominant LoS components, which require specialized joint MMSE estimation frameworks capable of simultaneously handling heterogeneous terrestrial and non-terrestrial channel statistics. These channel-related issues, in turn, render conventional distance-based user association strategies inadequate, since the distinct power profiles of satellite feeder links require a joint optimization of binary association decisions and continuous power control across both infrastructure layers to maximize system-wide energy efficiency. Yet, realizing the full potential of such systems requires a joint optimization framework that simultaneously addresses heterogeneous channel estimation, link budget asymmetry, and coupled association-power control. This is a timely and significant problem, given the accelerating push toward integrated terrestrial-satellite connectivity.

Evolutionary algorithms (EAs) constitute a class of population-based optimization techniques inspired by natural evolution, capable of addressing non-convex, non-differentiable, and highly nonlinear problems prevalent in NG network design \cite{he2023review,van2024active,vanchien2026singlemultiobjectivestochasticoptimization}. Their population-driven search enables efficient handling of computationally intensive and {even} NP-hard tasks, including throughput maximization under correlated fading, clustering optimization in cloud infrastructures, and adaptive routing for coverage enhancement. Among various EAs, Differential Evolution (DE) \cite{11080325, chien2024solving} has attracted significant attention owing to its simple parameterization, strong global search capability, and reliable convergence in high-dimensional continuous optimization, making it {eminently} suitable for large-scale {NG} resource management.

To the best of our knowledge, no existing work has jointly investigated energy efficiency of hybrid satellite-terrestrial cell-free Massive MIMO networks {in the face of realistic} imperfect channel state information (CSI) and user-centric association. For instance, \cite{van2022space} investigates space-terrestrial cooperation over spatially correlated channels with imperfect CSI for uplink performance analysis and optimization, but does not address energy-efficient operation under a user-centric Cell-Free Massive MIMO architecture nor a joint discrete-continuous optimization of association and power control; \cite{9136914} studies joint power allocation and load balancing for energy-efficient Cell-Free Massive MIMO in a purely terrestrial setting without incorporating satellites or LoS-dominated Rician fading channels; and \cite{hao2024joint} addresses joint user association and power control for terrestrial Cell-Free Massive MIMO but neither integrates LEO satellites nor targets energy efficiency under hybrid space-terrestrial coverage. As further confirmed by Table~\ref{tab:comparedwork}, our work uniquely combines: $i)$ multiple LEO satellites with spatially correlated Rician fading channels alongside distributed terrestrial APs; $ii)$ a closed-form uplink ergodic throughput expression under MRC that accounts for both LoS and NLoS components; $iii)$ a joint mixed-integer optimization of binary association and continuous power control; and $iv)$ an IDE algorithm with proven convergence for the resulting NP-hard problem. By exploiting these cooperative gains, the main contributions of this work are summarized as follows:
\begin{itemize}
    \item[$i)$] We derive closed-form expressions for the ergodic {uplink} throughput {of the MRC receiver} technique, explicitly characterizing the impact of both LoS and NLoS satellite components as well as {of} terrestrial fading statistics.
    \item[$ii)$] We introduce an energy efficiency maximization problem integrating user association and transmit power allocation, ensuring balanced throughput across heterogeneous users while accounting for binary {user-association} constraints.
    \item[$iii)$] To cope with the NP-hard nature of these problems, we develop a IDE framework, capable of efficiently exploring feasible {user-}association patterns and achieving near-optimal solutions in polynomial time.
    \item[$iv)$] Our numerical results demonstrate that the proposed approach improves the total energy efficiency compared to conventional fixed-association schemes. Moreover, the benefits of satellite assistance become particularly pronounced for users {having} unfavorable channel conditions.

\end{itemize}
The rest of this paper is organized as follows: Section \ref{Sec:Sys} introduces the system model, pilot training, and channel estimation protocol. Section \ref{Sec:Uplinkdata} presents the uplink throughput analysis. In addition, Section \ref{sec:problem} formulates the energy efficiency optimization problem {considered}, while Section \ref{sec:alg} describes our proposed IDE framework to acquire the problem. Finally, Section \ref{sec:result} provides {our} numerical results, followed by concluding remarks in Section \ref{sec:Conclusion}. 


\textit{Notations:} Matrices, vectors, scalars, and sets are denoted by upper-case boldface (e.g., $\mathbf{X}$), lower-case boldface (e.g., $\mathbf{x}$), italic (e.g., $x, X$), and calligraphic (e.g., $\mathcal{X}$) letters, respectively. $\mathcal{CN}(\cdot,\cdot)$, $\mathcal{N}(\cdot,\cdot)$, $\mathcal{U}(\cdot,\cdot)$, and $\mathcal{C}(\cdot,\cdot)$ denote the circularly symmetric Gaussian, Normal, Uniform, and Cauchy distributions, respectively. $(\cdot)^T$, $(\cdot)^H$, and $(\cdot)^{-1}$ denote the transpose, Hermitian transpose, and matrix inverse, respectively. Finally, $\mathrm{tr}(\mathbf{X})$, $\mathbf{I}_N$, and $\mathbf{0}_M$ denote the trace of $\mathbf{X}$, the $N \times N$ identity matrix, and the $M \times 1$ zero vector, respectively. 

\section{Space-Terrestrial System, Pilot Training, and Channel Estimation} \label{Sec:Sys}

This section delves into space-ground communications relying on multiple serving nodes (satellites and APs), addressing the practical challenges posed by realistic CSI and limited user associations.

\begin{figure*}[t]
    \centering
    \subfloat[]{\includegraphics[width=0.45\textwidth]{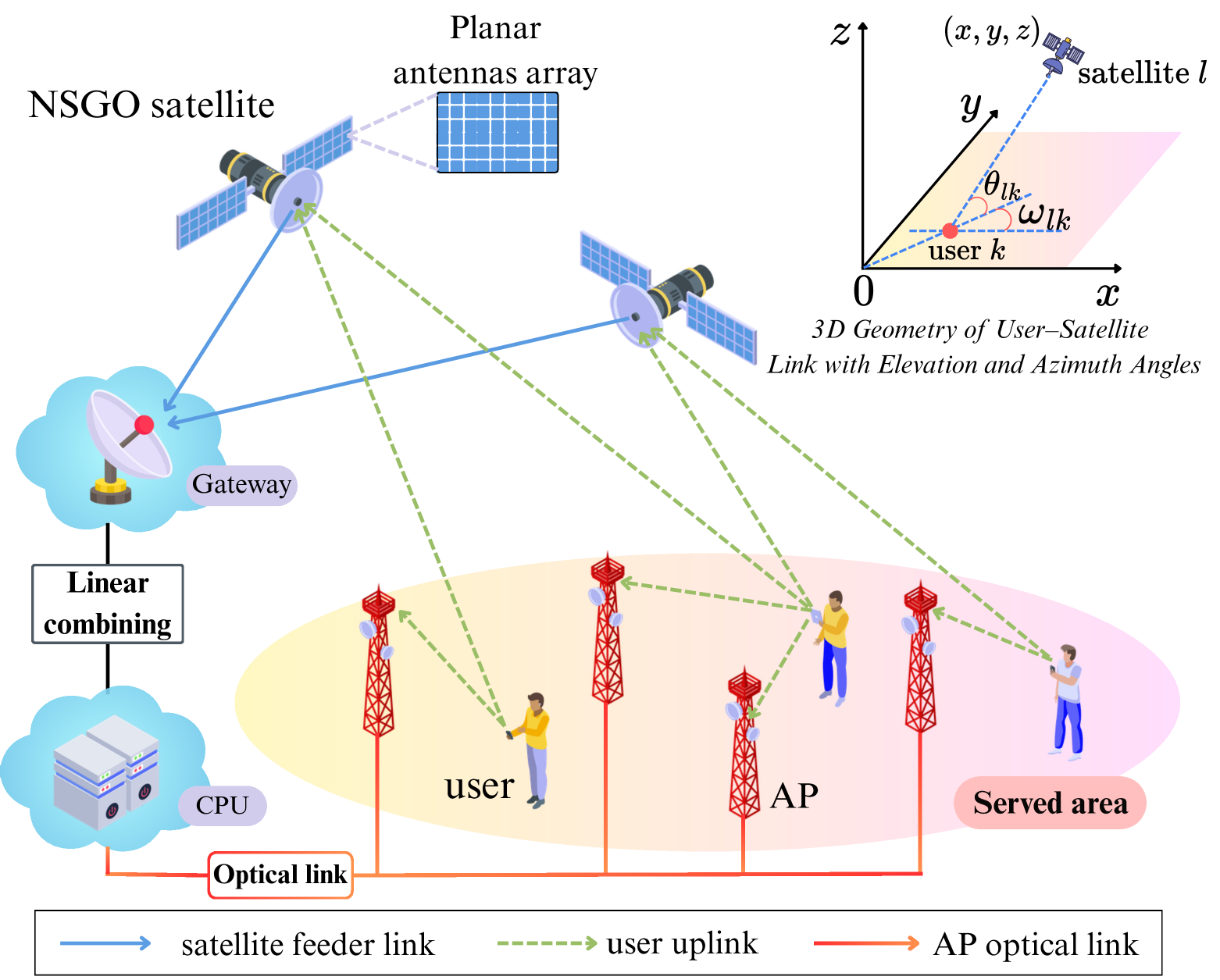}\label{fig:model}} \hspace{4mm}
    \subfloat[]{\includegraphics[width=0.42\textwidth]{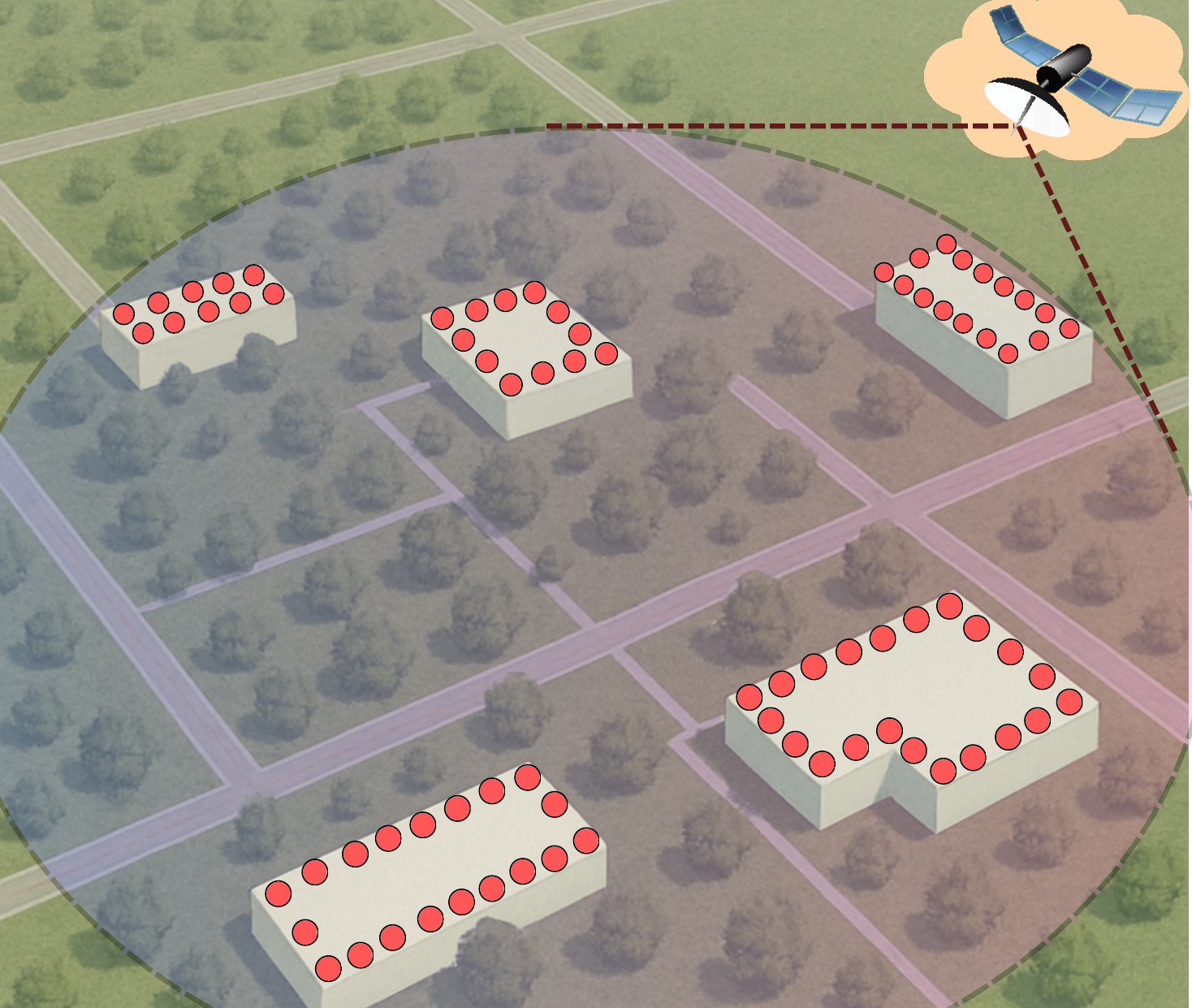}}
    \caption{The illustration of integrated Satellite-Cell-Free Massive MIMO architecture and dense AP deployment: (a) System architecture with satellite feeder links, user uplinks, AP optical links, and CPU processing; (b) An example of dense AP deployment (red dots) in a university campus.}
    \label{fig:Problem19}
\end{figure*}

\subsection{Proposed System Models}
We consider an integrated space-terrestrial network comprising $L$ LEO satellites in $\mathcal{M}^{\mathrm{SAT}}$ with $M$ antennas each, and $N$ terrestrial APs in $\mathcal{M}^{\mathrm{AP}}$ with $P$ antennas each, jointly serving $K$ single-antenna users in $\mathcal{M}^{\mathrm{user}}$. A flexible association scheme allocates each user to satellites or APs to improve load balancing and spectral efficiency. The system operates in the Massive MIMO regime having $ML + NP \gg K$, ensuring favorable propagation. As shown in Fig.~\ref{fig:model}, APs forward uplink signals to the CPU via optical fronthaul links~\cite{10546919}, while satellites use a radio feeder link through a ground station. Both links are modeled as imperfect Gaussian channels.
A quasi-static block-fading model is assumed, where channels remain constant over each coherence interval of $\tau_c$ symbols. Of these, $K$ symbols are allocated for pilot transmission and $\tau_c - K$ for uplink data~\footnote{\textcolor{black}{The coherence time is defined as $T_c(t) = \min \{\tau\mid| A_h(t, \tau)/A_h(t, 0)|< \epsilon \}$, where $A_h(t, \tau) = \mathbb{E}[h(t)h^H(t + \tau)]$ depends on both the spatial geometry and satellite dynamics~\cite{11104425}. In LEO systems, high mobility mainly introduces a deterministic Doppler shift, which can be largely removed through Doppler pre-compensation. Consequently, the residual Doppler spread is small, resulting in a sufficiently long coherence block $\tau_c \gg K$ that keeps the pre-log factor $(1 - K/\tau_c)$ in \eqref{eq:Rkv1} close to 1. Without compensation, the reduced coherence time would significantly degrade throughput. \textcolor{black}{For example, for a 300 km LEO orbit at 20 GHz, Doppler pre-compensation can by and large remove the dominant deterministic Doppler component. Assuming a representative residual Doppler frequency of approximately 100 Hz, the coherence time is $T_c \approx 0.423/f_{D,res} \approx 4.2$ ms~\cite{Rappaport_2024}, corresponding to about $42$ OFDM symbols for $\Delta f=10$ kHz. At low elevation angles, the higher absolute Doppler shift (which peaks near the horizon) amplifies the effect of ephemeris/prediction errors, while increased atmospheric scintillation and ground multipath increase the Doppler spread; both effects aggravate  the residual Doppler and thereby reduce the effective coherence time~\cite{11230625}.}}}. Then we adopt a user-centric clustering strategy where serving APs and satellites are selected based on each user's channel conditions. Following the dynamic cooperation clustering (DCC) framework of~\cite{demir2021foundations}, terrestrial or non-terrestrial operation is flexibly selected through optimized load balancing. User $k$ is served by APs in $\mathcal{M}_k^{\mathrm{AP}} \subseteq \{1,\ldots,N\}$ and satellites in $\mathcal{M}_k^{\mathrm{SAT}} \subseteq \{1,\ldots,L\}$. The associations are specified by binary variables $\alpha_{nk}^{\mathrm{AP}}, \alpha_{lk}^{\mathrm{SAT}} \in \{0,1\}$: 
\begin{equation}
\alpha^{\mathrm{AP}}_{nk} = 
\begin{cases}
1, & \text{if } n \in \mathcal{M}^{\mathrm{AP}}_k, \\
0, & \text{otherwise},
\end{cases}
 \text{ and }
\alpha^{\mathrm{SAT}}_{lk} = 
\begin{cases}
1, & \text{if } l \in \mathcal{M}^{\mathrm{sat}}_k, \\
0, & \text{otherwise}.
\end{cases}
\end{equation}
By controlling these variables, the system flexibly adapts AP and satellite assignments, enabling joint optimization of spectral efficiency, load balancing, and user-centric service.
\begin{remark}
Although uplink (UL) signals received by the satellite are weaker than those at terrestrial APs, the large antenna array at the satellite compensates through high receiver gain. Coherent combining at the CPU then leverages this contribution to enhance overall link quality, an effect that becomes more pronounced when both the satellite and ground station employ high-gain arrays, effectively mitigating the severe space-ground path loss \cite{8746876}.
\end{remark}
\begin{remark}
LEO satellites are vital for sparse rural or disaster-prone regions lacking reliable terrestrial infrastructure. In such scenarios, our optimization framework inherently adapts to satellite-only mode by driving $\alpha_{nk}^{\mathrm{AP}}$ to $0$. While the hybrid architecture maximizes energy efficiency in dense environments, this seamless fallback ensures ubiquitous connectivity and demonstrates the system's architectural flexibility.
\end{remark}
\subsection{Channel Model}\label{sec:channel model}
When AP~$n$ serves user~$k$, the associated terrestrial uplink channel vector is denoted by $\mathbf{g}_{nk} \in \mathbb{C}^P$ and modeled as:
\begin{equation}
    \mathbf{g}_{nk} \sim \mathcal{CN}(\mathbf{0}_P, \mathbf{G}_{nk}),
\end{equation}
where $\mathbf{G}_{nk} = \beta^{\mathrm{AP}}_{nk} \mathbf{I}_P$, and $\beta^{\mathrm{AP}}_{nk} \ge 0$ is the large-scale fading coefficient capturing geometric pathloss and shadow fading.  
For the relatively small number of satellites, spatial correlation is modeled by a tractable correlated Rayleigh fading model. The spatial channel \(\mathbf{h}_{lk} \in \mathbb{C}^M\) between user \(k\) and satellite \(l\) is given by \(\mathbf{h}_{lk} \sim \mathcal{CN}(\bar{\mathbf{h}}_{lk}, \mathbf{H}_{lk})\), where \(\bar{\mathbf{h}}_{lk} \in \mathbb{C}^M\) is the LoS component and \(\mathbb{E}\{\mathbf{h}_{lk}\mathbf{h}^H_{lk}\} = \mathbf{H}_{lk} \in \mathbb{C}^{M \times M}\) describes the spatial correlation. The LoS component is expressed as:
\begin{equation}\label{eq:gnk}
    \bar{\mathbf{h}}_{lk} = \sqrt{\beta^{\mathrm{SAT}}_{lk} \frac{\eta_{lk}}{(\eta_{lk} + 1)}} \left[ e^{j\pmb{\ell}(\theta_{lk}, \omega_{lk})^T\mathbf{c}_1},\ldots,e^{j\pmb{\ell}(\theta_{lk}, \omega_{lk})^T\mathbf{c}_M}  \right]^T,
\end{equation}

where $\theta_{lk}$ and $\omega_{lk}$ denote the elevation and azimuth angles, respectively; $\eta_{lk} \geq 0$ is the Rician factor characterizing the power ratio between {the} LoS and NLoS components; and $\beta^{\mathrm{SAT}}_{lk}$ represents the large-scale fading coefficient between user~$k$ and satellite~$l$, which is influenced by both the satellite's altitude and the user's geographical position. We consider a rectangular antenna array configuration (as illustrated in Fig. \ref{fig:model}), for which the wavefront direction vector $\pmb{\ell}(\theta_{lk}, \omega_{lk})$ is characterized by $\pmb{\ell}(\theta_{lk}, \omega_{lk}) = \frac{2\pi}{\lambda} \left[\cos(\theta_{lk}) \cos(\omega_{lk}), \sin(\theta_{lk}) \cos(\omega_{lk}), \sin(\theta_{lk})
\right]^T,$ where $\lambda$ represents the carrier wavelength. The vectors {$\mathbf{c}_m$ for $m=\{1,\ldots,M\}$} correspond to the $M$ indexing vectors as detailed in \cite{van2022space}. Additionally, $d_U$ and $d_V$ specify the antenna element spacing in the horizontal and vertical dimensions, respectively. Following the 3D channel modeling approach presented in \cite{ying2014kronecker} and \cite{5288965}, the spatial correlation matrix for the planar antenna array is formulated as:
\begin{equation}\label{eq:Gnk}
    \mathbf{H}_{lk} = \frac{\beta^{\mathrm{SAT}}_{lk}}{\eta_{lk} + 1} \mathbf{H}_{lk,U} \otimes \mathbf{H}_{lk,V},
\end{equation}
where $\mathbf{H}_{lk,U} \in \mathbb{C}^{M_U \times M_U}$ and $\mathbf{H}_{lk,V} \in \mathbb{C}^{M_V \times M_V}$ characterize {the} spatial correlation along the horizontal and vertical dimensions, respectively.

\subsection{Uplink Pilot Training Phase}
During the uplink pilot training phase, all $K$ users transmit their respective pilot sequences simultaneously to the network nodes. 
User~$k$ transmits its assigned pilot sequence $\boldsymbol{\phi}_k \in \mathbb{C}^{K}$, designed to be mutually orthogonal, satisfying $\pmb{\phi}_k^H \pmb{\phi}_{k'} = 1 \text{ if } k = k'$ and $\pmb{\phi}_k^H \pmb{\phi}_{k'} = 0,$ otherwise.
The pilot signal received at the AP~$n$, denoted by $\mathbf{Y}^{\text{pilot}}_{n} \in \mathbb{C}^{P \times K}$, and at the satellite~$l$, denoted by $\tilde{\mathbf{Y}}^{\text{pilot}}_{l} \in \mathbb{C}^{M \times K}$, can be respectively expressed as:
\begin{align}
&\mathbf{Y}^{\mathrm{AP},\text{ pilot}}_{n} = \sum\nolimits_{k=1}^K \alpha^{\mathrm{AP}}_{nk}\sqrt{\rho K}  \mathbf{g}_{nk} \pmb{\phi}_k^H + \mathbf{N}^{\mathrm{AP}}_{n}, \label{eq:ypm}\\
&{\mathbf{Y}}^{\mathrm{SAT},\text{ pilot}}_{l} = \sum\nolimits_{k=1}^K {\alpha}^{\mathrm{SAT}}_{lk}\sqrt{\rho K}  \mathbf{h}_{lk} \pmb{\phi}_k^H + {\mathbf{N}}^{\mathrm{SAT}}_{l}, \label{eq:Yp}
\end{align}
where $\rho \geq 0$ represents the transmit power allocated per pilot symbol. The {signals received} at AP~$n$ and satellite~$l$ are corrupted by additive white Gaussian noise (AWGN), represented by the matrices $\mathbf{N}^{\mathrm{AP}}_{n} \in \mathbb{C}^{P\times K}$ and ${\mathbf{N}}^{\mathrm{SAT}}_{l} \in \mathbb{C}^{M \times K}$, respectively. These noise components are characterized by complex Gaussian distributions with zero mean, specifically $\mathbf{N}^{\mathrm{AP}}_{n} \sim \mathcal{CN}(\mathbf{0}_P, \sigma_a^2\mathbf{I}_P)$ for the terrestrial AP~$n$ and $\mathbf{N}^{\mathrm{SAT}}_{l} \sim \mathcal{CN}(\mathbf{0}_M, \sigma_s^2\mathbf{I}_M)$ for the satellite~$l$, where the parameters $\sigma_a^2 [\mathrm{dB}]$ and $\sigma_s^2 [\mathrm{dB}]$ denote the respective noise variances at the AP and satellite reception points.

\subsection{Uplink Channel Estimation Phase}
From the received pilot signals in \eqref{eq:ypm} and \eqref{eq:Yp}, and assuming that the APs and satellites know the relevant channel statistics, MMSE estimation is applied to obtain accurate estimates of $\mathbf{g}_{nk}$ and $\mathbf{h}_{lk}$, as stated in the lemma below.
\begin{lemma} \label{lemma:ChannelEst}
When employing the MMSE estimation technique, the terrestrial channel vector $\mathbf{g}_{nk}$ can be estimated from the received pilot signal as:
\begin{equation}
\hat{\mathbf{g}}_{nk} = {\sqrt{\rho K} \beta^{\mathrm{AP}}_{nk}}{(\rho K  \beta^{\mathrm{AP}}_{nk} + \sigma_{a}^2)^{-1}}\mathbf{Y}_n^{\mathrm{AP}, \text{pilot}}\pmb{\phi}_k.
\end{equation}
The channel estimates $\hat{\mathbf{g}}_{nk}$ are independent random variables, which are distributed as $\hat{\mathbf{g}}_{nk} \sim \mathcal{CN}
\left( \mathbf{0}_P, \;
\rho K\,\mathbf{G}_{nk}\mathbf{\Psi}_{nk}\mathbf{G}_{nk} \right) 
\sim \mathcal{CN} \left(
\mathbf{0}_P, {\rho K \left(\beta^{\mathrm{AP}}_{nk}\right)^2}(\rho K\,\beta^{\mathrm{AP}}_{nk} + \sigma_a^2)^{-1} \mathbf{I}_P \right),$ and the variance is $\mathbf{\Theta}_{nk} =\mathbb{E}\{\mathbf{|\hat{g}}_{nk}|^2\}= {\rho K \left(\beta^{\mathrm{AP}}_{nk}\right)^2}(\rho K\,\beta^{\mathrm{AP}}_{nk} + \sigma_a^2)^{-1} \mathbf{I}_P,$
where $\mathbf{\Psi}_{nk} = ({\rho K}\mathbf{G}_{nk} + \sigma_a^2\mathbf{I}_P)^{-1}$ denotes the MMSE coefficient matrix. The channel estimation error, $\mathbf{e}^\mathrm{AP}_{nk} = \mathbf{g}_{nk} - \hat{\mathbf{g}}_{nk}$, {obeys} a complex Gaussian distribution $\mathbf{e}^\mathrm{AP}_{nk} \sim \mathcal{CN}(\mathbf{0}_P,\mathbf{G}_{nk}-\rho K\mathbf{G}_{nk}\mathbf{\Psi}_{nk}\mathbf{G}_{nk}) \sim \mathcal{CN}(\mathbf{0}_P,{\beta^{\mathrm{AP}}_{nk} \sigma_a^2}{(\rho K \beta^{\mathrm{AP}}_{nk} + \sigma_a^2)^{-1}}  \mathbf{I}_P) \sim \mathcal{CN}(\mathbf{0}_P,\mathbf{C}_{e,nk}^\mathrm{AP})$ with $\mathbf{C}_{e,nk}^\mathrm{AP}$ {being} the correlation matrix.

For the {LEO} channel $\mathbf{h}_{lk}$, the MMSE estimate $\hat{\mathbf{h}}_{lk}$ is: 
\begin{equation} \label{eq:hatk}
\hat{\mathbf{h}}_{lk} = \bar{\mathbf{h}}_{lk} + \sqrt{\rho K}\mathbf{H}_{lk} \pmb{\Phi}_{lk} ( {\mathbf{Y}}^{\mathrm{SAT}, \text{pilot}}_{l} \pmb{\phi}_{k} - \sqrt{\rho K} \bar{\mathbf{h}}_{lk} ),
\end{equation}
where $\pmb{\Phi}_{lk} = (\rho K \mathbf{H}_{lk} + \sigma_s^2 \mathbf{I}_M )^{-1}$. Since $\mathbf{\Omega}_{lk} = \mathbf{H}_{lk} \pmb{\Phi}_{lk}\mathbf{H}_{lk}$, the distribution of the channel estimate is $\hat{\mathbf{h}}_{lk} \sim \mathcal{CN}( \bar{\mathbf{h}}_{lk}, \rho K \mathbf{\Omega}_{lk})$. Furthermore, the estimation error $\mathbf{e}^\mathrm{SAT}_{lk} = \mathbf{h}_{lk} - \hat{\mathbf{h}}_{lk}$ is distributed as $\mathbf{e}^\mathrm{SAT}_{lk} \sim \mathcal{CN}(\mathbf{0}, \mathbf{H}_{lk} - \rho K \mathbf{\Omega}_{lk}) \sim \mathcal{CN}(\mathbf{0}, \mathbf{C}_{e,lk}^\mathrm{SAT})$ as its correlation matrix is $\mathbf{C}_{e,lk}^\mathrm{SAT}$.
\end{lemma}
\begin{proof}
The proof is obtained by applying the standard MMSE estimation framework of \cite{Kay1993a}, as detailed in Appendix~\ref{Apx:ChannelEst}.
\end{proof}

\section{Uplink Data Transmission and Ergodic Throughput Analysis} \label{Sec:Uplinkdata}
\subsection{Uplink Data Transmission}
During the uplink data transmission phase, all $K$ users simultaneously transmit signals to their serving APs and/or satellites. {User} $k$ transmits a data symbol $s_k$ {at} unit power, i.e., $\mathbb{E}\{|s_k|^2\} = 1$. The signal received at AP $n$, denoted by $\mathbf{y}_n \in \mathbb{C}^P$, can be expressed as: 
\begin{equation} \label{eq:ym}
\mathbf{y}^{\mathrm{AP}}_{n} = \sum\nolimits_{k=1}^K \alpha^{\mathrm{AP}}_{nk}\sqrt{\rho_k}\mathbf{g}_{nk} s_k + \mathbf{n}^{\mathrm{AP}}_n,
\end{equation}
where $\rho_k>0$ is the transmit power allocated by user $k$ to its data symbol, and $\mathbf{n}^\mathrm{AP}_n \sim \mathcal{CN}(\mathbf{0}_P, \sigma_n^2 \mathbf{I}_P)$ is the {AWGN}. Similarly, at the gateway of satellite $l$, the received signal $\mathbf{y}_l \in \mathbb{C}^M$ is the superposition of signals from all active users: 
 \begin{equation} \label{eq:y}
{\mathbf{y}}^{\mathrm{SAT}}_l = \sum\nolimits_{k=1}^K {\alpha}^{\mathrm{SAT}}_{lk} \sqrt{\rho_k} \mathbf{h}_{lk} s_k + {\mathbf{n}}^{\mathrm{SAT}}_l,
\end{equation}
where $\mathbf{n}^\mathrm{SAT}_l \sim \mathcal{CN}(\mathbf{0}_M, \sigma_l^2 \mathbf{I}_M)$ is the additive noise during the uplink data transmission. At the CPU, the signal sent from user~$k$ is decoded by the following combination:
\begin{equation}
    \begin{aligned} \label{eq:hatsk}
\hat{s}_k = &\sum\nolimits_{n=1}^N \alpha^{\mathrm{AP}}_{nk}\hat{s}_{nk}+\sum\nolimits_{l=1}^L \alpha^{\mathrm{SAT}}_{lk}\hat{s}_{lk} \\
            = & \sum\nolimits_{n=1}^N \alpha^{\mathrm{AP}}_{nk}(\mathbf{w}^{\mathrm{AP}}_{nk})^H {\mathbf{y}^\mathrm{AP}_n} + \sum\nolimits_{l=1}^L \alpha^{\mathrm{SAT}}_{lk}({\mathbf{w}^{\mathrm{SAT}}_{lk}})^H \mathbf{y}^{\mathrm{SAT}}_{l},
\end{aligned}
\end{equation}
where $\hat{s}_{lk}$ and $\hat{s}_{nk}$ are the estimate of $s_k$ at the satellite~$l$ and the AP~$n$. {Furthermore,} $\mathbf{w}_{lk} \in \mathbb{C}^N$ and $\mathbf{w}_{nk} \in \mathbb{C}^P$ are centralized combining vectors utilized at {the} gateway of satellite~$l$ and AP~$n$, respectively. Since either $\alpha^{\mathrm{AP}}_{nk} = 0$ or $\alpha^{\mathrm{SAT}}_{lk} = 0$ implies that the corresponding {receiver combining} vectors $\alpha^{\mathrm{AP}}_{nk}(\mathbf{w}^{\mathrm{AP}}_{nk})^H$ and $\alpha^{\mathrm{SAT}}_{lk}({\mathbf{w}^{\mathrm{SAT}}_{lk}})^H$ vanish in \eqref{eq:hatsk}, the CPU effectively performs receive combining using only the data signals {gleamed} from the APs and satellites {having} $\alpha^{\mathrm{AP}}_{nk} \neq 0$ and $\alpha^{\mathrm{SAT}}_{lk} \neq 0$. Hence, \eqref{eq:hatsk} can equivalently be rewritten as:
\begin{equation}
\hat{s}_k = \sum\nolimits_{n \in \mathcal{M}^{\mathrm{AP}}_{k}}^{} \mathbf{w}^H_{nk} \mathbf{y}^{\mathrm{AP}}_n + \sum\nolimits_{l \in \mathcal{M}^{\mathrm{SAT}}_{k}} \mathbf{w}^H_{lk} \mathbf{y}^{\mathrm{SAT}}_l,
\end{equation}
where the summation is restricted to the subset of APs and satellites that actively serve user~$k$. {Albeit} this form is more compact and offers structural insights, it would lead to cumbersome notation during {our} performance analysis. Therefore, we retain the original expression in \eqref{eq:hatsk} throughout this work.
\subsection{Ergodic Throughput Analysis}
In the most advanced uplink realization of Cell-Free Massive MIMO, the operation is fully centralized. In this architecture, the APs {simply act} as remote radio heads, forwarding their received baseband signals to the CPU for joint processing \cite{demir2021foundations}.
To {detect} the data symbol $s_k$ {transmitted} from user $k$, we utilize a processing approach where {the} signals {impinging} from multiple APs and satellites are {coherently} combined. By substituting the formulations of the received signals in \eqref{eq:ym} and \eqref{eq:y} into \eqref{eq:hatsk}, {we arrive at:}
\begin{multline} \label{eq:shatk}
\hat{s}_k = \sum\limits_{k'=1}^K \sqrt{\rho_{k'}} \left(\sum\limits_{l=1}^L{\alpha}^\mathrm{SAT}_{lk'}  (\mathbf{w}^{\mathrm{SAT}}_{lk})^H \mathbf{h}_{lk'} + \sum\limits_{n=1}^N {\alpha}_{nk'}^\mathrm{AP} \mathbf{w}_{nk}^H \mathbf{g}_{nk'} \right) s_{k'} \\ + \sum\nolimits_{l=1}^L  {\alpha}_{lk}^\mathrm{SAT} (\mathbf{w}^{\mathrm{SAT}}_{lk})^H \mathbf{n}^\mathrm{SAT}_l + \sum\nolimits_{n=1}^N \alpha_{nk}^\mathrm{AP} \mathbf{w}^H_{nk} \mathbf{n}^\mathrm{AP}_n,
\end{multline}
which represents a composite signal comprising contributions from multiple transmitters {destined to} different receivers, along with additive noise. In \eqref{eq:shatk}, we observe that effective load balancing helps in suppressing {the} inter-user interference, whereas the additive noise originates collectively from all APs and satellites. To facilitate the analysis, we now define a new variable as:
\begin{equation} \label{eq:okkprime}
o_{kk'} = \sum\nolimits_{l=1}^L {\alpha}^\mathrm{SAT}_{lk'}  (\mathbf{w}^{\mathrm{SAT}}_{lk})^H \mathbf{h}_{lk'} + \sum\nolimits_{n=1}^N {\alpha}^\mathrm{AP}_{nk'} (\mathbf{w}^\mathrm{AP}_{nk} )^H\mathbf{g}_{nk'}, 
\end{equation}
which stands for the overall channel, taking the user-association into account. Furthermore, let us denote the aggregated noise as:
\begin{equation} \label{eq:ntildek}
z_k = \sum\nolimits_{l=1}^L  {\alpha}^\mathrm{SAT}_{lk} {\mathbf{w}}_{lk}^H {\mathbf{n}}_l^\mathrm{SAT} + \sum\nolimits_{n=1}^N \alpha_{nk}^\mathrm{AP} \mathbf{w}^H_{nk} \mathbf{n}_n^\mathrm{AP}.
\end{equation}
Then by substituting \eqref{eq:okkprime} and \eqref{eq:ntildek} into \eqref{eq:shatk}, the signal of user~$k$ {detected} at the CPU is {expressed} as:
\begin{equation}\label{eq:hatskv2}
    \hat{s}_k = \sum\nolimits_{k'=1}^K \sqrt{\rho_{k'}} o_{kk'}s_{k'} + z_k.
\end{equation}
Then the expression {in} \eqref{eq:hatskv2} can be reformulated as in \eqref{eq:hatskv1}, where the first part in the second equation of \eqref{eq:hatskv1} represents the desired signal {impinging} from user~$k$ with a deterministic channel gain; the second term corresponds to the self-interference caused by channel estimation errors; the third term accounts for the mutual interference from other users; and the last term represents the additive noise:
\begin{figure*}
    \begin{equation} \label{eq:hatskv1}
 \hat{s}_k =   \underbrace{\sqrt{\rho_k} \mathbb{E}\{o_{kk} \} s_{k}}_{\text{Desired signal over deterministic channel gain}}+ 
 \underbrace{\sqrt{\rho_k} ( o_{kk} - \mathbb{E}\{o_{kk} \} ) s_{k}}_{\text{Self-interference}} + \underbrace{\sum\nolimits_{\substack{k'=1 \\ k' \neq k}}^K \sqrt{\rho_{k'}} o_{kk'} s_{k'}}_{\text{Interference}} +   \underbrace{z_k}_{\text{Noise}}.
\end{equation}
\hrule
\end{figure*}
By exploiting the use-and-then-forget channel capacity bounding technique, the {ergodic uplink} throughput of user~$k$ becomes:
\begin{equation} \label{eq:Rkv1}
R_k = B\left(1 - {K}/{\tau_c}\right) \log_2 (1 + \mathrm{SINR}_k), \mbox{ [Mbps]},
\end{equation}
where $B$~[MHz] is the system bandwidth and the effective signal-to-interference-and-noise ratio (SINR) is given in \eqref{eq:SINRk}. The {ergodic uplink} throughput in \eqref{eq:Rkv1} can be applied to an arbitrary channel model and detection vectors by numerically evaluating several expectations in the numerator and denominator of \eqref{eq:SINRk} over numerous realizations of the small-scale fading coefficients.

\subsection{Linear Combining Techniques}\label{seec:combtech}
In the following, we introduce several linear {RC} schemes that can be employed for {reception} in cell-free massive MIMO systems \cite{bashar2020uplink,jiang2025most, demir2021foundations}. Let us consider two distinct transmission scenarios denoted by $\xi \in \{\mathrm{AP}, \mathrm{SAT}\}$, where $\mathrm{AP}$ and $\mathrm{SAT}$ represent the terrestrial AP and satellite links, respectively. We define $Q$ as the total number of serving nodes and $Q^{(\xi)}$ as the number of antennas deployed at each serving node under scenario~$\xi$. Accordingly, the {RC} matrix of user~$k$ for scenario~$\xi$ is expressed as:
\begin{equation}
    \mathbf{W}^{(\xi)}_{k} = \big[ \mathbf{w}^{(\xi)}_{1k}, \ldots, \mathbf{w}^{(\xi)}_{Qk} \big] \in \mathbb{C}^{Q \times Q^{(\xi)}},
\end{equation}
where $\mathbf{W}^{(\xi)}_{k}$ denotes the linear MMSE combining matrix configured for the $\xi$-link, which is defined as:
\begin{equation}\label{eq:MRC_RZF}
\mathbf{{W}}^{(\xi)}_k =
\begin{cases}
\hat{\mathbf{R}}^{(\xi)}_{k}, &\text{for MRC}, \\[5pt]
\hat{\mathbf{R}}^{(\xi)}_{k}\left(( {\hat{\mathbf{R}}^{(\xi)}_{k}})^{H}\hat{\mathbf{R}}^{(\xi)}_{k}+{(\sigma^{(\xi)})}^2 \mathbf{I}^{(\xi)}\right)^{-1},  &\text{for RZF}, \
\end{cases}
\end{equation}
in case MMSE combining, the element vector $\mathbf{w}_{qk}^{(\xi)}$ is:
\begin{equation}\label{eq:MMSE}
  \mathbf{w}_{qk}^{(\xi)} = 
\left(\eta{\sum\nolimits_{k'\in\mathcal{M}^{(\xi)}_k}\hat{\mathbf{r}}_{qk'}} \hat{\mathbf{r}}_{qk'}^{H}
    + \eta\mathbf{C}^{(\xi)}_{e,qk}+{(\sigma^{(\xi)})}^2 \mathbf{I}^{(\xi)}
\right)^{-1}
\hat{\mathbf{r}}_{qk}
\end{equation}
In the terrestrial scenario, where users are served by the APs, the corresponding {RC} matrices and system parameters are specified as follows: $\xi = \mathrm{AP}$, we have 
$\hat{\mathbf{R}}^{\mathrm{AP}}_k=\hat{\mathbf{G}}_k = [\mathbf{\widehat{g}}_{1k}, \ldots, \mathbf{\widehat{g}}_{nK}] \in \mathbb{C}^{P \times N}, \mathbf{W}_k^{(\xi)}=\mathbf{W}^\mathrm{AP}_k  \text{ and, } \mathcal{M}_k^{(\xi)}=\mathcal{M}_k^\mathrm{AP}.$ Then, the element vector is $\mathbf{w}^{\mathrm{AP}}_{qk} = \mathbf{w}^{\mathrm{AP}}_{nk}.$ Moreover, $
\sigma^{\mathrm{AP}}=\sigma_n, $ and $\mathbf{I}^{\mathrm{AP}}=\mathbf{I}_P,
$ with all $k\in \mathcal{M}^{\mathrm{user}}$. {By contrast,} in the satellite scenario, the users are connected through the space segment, {where} the corresponding {RC} matrices and parameters are defined accordingly at the satellite gateway. {For} $\xi = \mathrm{SAT}$, these are specified {as} $\hat{\mathbf{R}}^{\mathrm{SAT}}_k=\hat{\mathbf{H}}_k = [\mathbf{\widehat{h}}_{1k}, \ldots, \mathbf{\widehat{h}}_{lk}] \in \mathbb{C}^{M \times L}, \mathbf{W}_k^{(k)} =\mathbf{W}_k^\mathrm{SAT} = \{\mathbf{w}^{\mathrm{SAT}}_{1k}, \cdots, \mathbf{w}^{\mathrm{SAT}}_{Lk} \}, 
\sigma^{\mathrm{SAT}}=\sigma_l,$ and $ \mathbf{I}^{\mathrm{SAT}}=\mathbf{I}_M,
$ with all $k\in \mathcal{M}^{\mathrm{user}}$.

For {the MRC}, the total detection $\sum_{n\in \mathcal{M}^\mathrm{AP}_k} \hat{\mathbf{g}}_{nk}^H\mathbf{y}_{nk}^\mathrm{AP}$ requires $|\mathcal{M}^\mathrm{AP}_k|P$ multiplications. Hence, for all $K$ users, the {total} complexity is $\mathcal{O}\left(P\sum_{k=1}^K|\mathcal{M}^\mathrm{AP}_k|\right)$ per channel use.
For regularized zero-forcing (RZF) and MMSE combining, the complexity is dominated by matrix multiplications and inversions, as defined in \eqref{eq:MRC_RZF}-\eqref{eq:MMSE}. For RZF, three major operations are required: $(i)$ computing $\hat{\mathbf{G}}_k^H\hat{\mathbf{G}}_k$ {at a complexity order of} $\mathcal{O} (|\mathcal{M}^\mathrm{AP}_k|^3 P)$;
$(ii)$ inverting $(\hat{\mathbf{G}}_k^H\hat{\mathbf{G}}_k + \sigma_n^2\mathbf{I}_N)$, which costs $\mathcal{O}(|\mathcal{M}^\mathrm{AP}_k|^2)$; and $(iii)$ multiplying $\hat{\mathbf{G}}_k$ by the inverted matrix, requiring $\mathcal{O}(P|\mathcal{M}^\mathrm{AP}_k|^2)$ {operations}. Thus, the total RZF complexity for all users is $\mathcal{O}\left(\sum_{k=1}^K\left(|\mathcal{M}^\mathrm{AP}_k|^3 + P|\mathcal{M}^\mathrm{AP}_k|^2\right)\right)$.
For MMSE combining, computing $\mathbf{w}_{nk}^\mathrm{AP}$ involves:
$(i)$ forming $\sum_{k'=1}^K\hat{\mathbf{g}}_{nk'}\hat{\mathbf{g}}_{nk'}^H$ {at a cost of} $\mathcal{O}(KP^2)$;
$(ii)$ matrix inversion with cost $\mathcal{O}(P^3)$; and
$(iii)$ matrix-vector multiplication {with a cost of} $\mathcal{O}(P^2)$.
Repeating this for all APs serving user $k$ yields a total complexity of $\mathcal{O}\left(|\mathcal{M}^\mathrm{AP}_k|(KP^2 + P^3)\right)$. Therefore, the overall MMSE complexity across all users is $\mathcal{O}\left(\left(\sum_{k=1}^K|\mathcal{M}^\mathrm{AP}_k|\right)KP^2 + \sum_{k=1}^K|\mathcal{M}^\mathrm{AP}_k|P^3\right)$.
Similarly, for satellite-side processing, the MRC, RZF, and MMSE schemes require total complexities of $\mathcal{O}\left(M\sum_{k=1}^K|\mathcal{M}^\mathrm{SAT}_k|\right)$, $\mathcal{O}\left(\sum_{k=1}^K\left(|\mathcal{M}^\mathrm{SAT}_k|^3 + M|\mathcal{M}^\mathrm{SAT}_k|^2\right)\right)$, and $\mathcal{O}\left(\left(\sum_{k=1}^K|\mathcal{M}^\mathrm{SAT}_k|\right)KM^2 + \sum_{k=1}^K|\mathcal{M}^\mathrm{SAT}_k|M^3\right)$, respectively.

MMSE and RZF are costly for large-scale networks {having} many APs and users. Therefore in this work, we derive an exact closed-form expression for the uplink ergodic rate \eqref{eq:Rkv1} under MRC detection as shown in Theorem~\ref{Theorem:SE}, which holds for arbitrary values of $M$, $N$, $P$, $L$, and $K$. The MRC scheme also provides a scalable analytical framework that facilitates network expansion {to} a large number of APs and users. Although it is possible to develop an approximate closed-form characterization of the {ergodic uplink} rate for partial MMSE detection, such an extension would require adopting a different analytical approach.

\begin{figure*}[t]
\begin{equation} \label{eq:SINRk}
\mathrm{SINR}_k = \frac{\rho_k |\mathbb{E}\{o_{kk} \} |^2 }{\sum_{k'=1}^K \rho_{k'} \mathbb{E}\{ |o_{kk'}|^2 \} - \rho_k|\mathbb{E}\{o_{kk} \} |^2 + \sum_{l=1}^L\mathbb{E}\{\|\mathbf{w}^{\mathrm{SAT}}_{lk}\|^2\} \sigma_s^2 + \sum_{n=1}^N \mathbb{E}\{\|\mathbf{w}_{nk}^{\mathrm{AP}}\|^2\} \sigma_a^2}
\end{equation}
\vspace{-0.8cm}
\end{figure*}
\begin{theorem} \label{Theorem:SE}
If {the} MRC technique is {employed} at the gateways and APs, the ergodic {uplink} throughput of user~$k$ is evaluated by \eqref{eq:Rkv1}, where the effective SINR is {evaluated} as:
{\small
\begin{equation} \label{eq:SINRmrc}
\mathrm{SINR}_k^{\mathrm{mrc}} = \frac{\rho_k \left(\sum\limits_{l=1}^L \alpha_{lk}^{\mathrm{SAT}}\left(\|\overline{\mathbf{h}}_{lk}\|^2 + \rho K \operatorname{tr}(\pmb{\Omega}_{lk})\right) + \sum\limits_{n=1}^N \alpha_{nk}^{\mathrm{AP}}\mathrm{tr}(\mathbf{\Theta}_{nk})\right)^2}{\mathsf{MI}_k (\{\alpha_{nk}^\mathrm{AP} \}, \{ {\alpha}^\mathrm{SAT}_{lk} \}) + \mathsf{NO}_k (\{\alpha^\mathrm{AP}_{nk} \}, \{ {\alpha}^\mathrm{SAT}_{lk} \})}
\end{equation}
} 

The mutual interference, denoted as $\mathsf{MI}_k$, and the noise, denoted as $\mathsf{NO}_k$, are respectively defined as follows:
\begin{figure*}
 \begin{equation}
    \begin{split}
        &\mathsf{MI}_k (\{\alpha^\mathrm{AP}_{nk} \}, \{ {\alpha}^\mathrm{SAT}_{lk} \}) = \ \sum\nolimits_{\substack{k'=1 \\ k' \neq k}}^K   \alpha^\mathrm{SAT}_{lk'}\rho_{k'}   \Big\lvert \sum\nolimits_{l=1}^L {\alpha}^\mathrm{SAT}_{lk}  \bar{\mathbf{h}}_{lk}^H \bar{\mathbf{h}}_{lk'}\Big \lvert^2  + \sum\nolimits_{l=1}^L {\alpha}^\mathrm{SAT}_{lk} \Big( \rho K \sum\nolimits_{k' =1}^K  \alpha^\mathrm{SAT}_{lk'}\rho_{k'}   \bar{\mathbf{h}}_{lk'}^H\pmb{\Omega}_{lk} \bar{\mathbf{h}}_{lk'} \\& + \sum\nolimits_{k'=1}^K {\alpha}^\mathrm{SAT}_{lk'} \rho_{k'} \bar{\mathbf{h}}_{lk}^H \mathbf{H}_{lk'} \bar{\mathbf{h}}_{lk}  + \rho K \sum\nolimits_{k' =1 }^K  \alpha^\mathrm{SAT}_{lk'} \rho_{k'}   \mathrm{tr}(\mathbf{H}_{lk'}\pmb{\Omega}_{lk}) \Big) + \sum\nolimits_{n=1}^N {\alpha}^\mathrm{AP}_{nk} \Big( \sum\nolimits_{k'=1}^K  {\alpha}^\mathrm{AP}_{nk'} \rho_{k'} \mathrm{tr}(\mathbf{G}_{nk'} \pmb{\Theta}_{nk})) \Big) \label{eq:MIk}
    \end{split}
\end{equation}

\begin{equation}
\mathsf{NO}_k (\{\alpha^\mathrm{AP}_{nk} \}, \{ {\alpha}^\mathrm{SAT}_{lk} \}) = \sigma_s^2
\sum\nolimits_{l=1}^L {\alpha}^\mathrm{SAT}_{lk}(\|\overline{\mathbf{h}}_{lk}\|^2 + \rho K \operatorname{tr}(\pmb{\Omega}_{lk}))
+\sigma_a^2
\sum\nolimits_{n=1}^N {\alpha}^\mathrm{AP}_{nk} \operatorname{tr}(\pmb{\Theta}_{nk}).
\end{equation}
\hrule
\end{figure*}
\end{theorem}
\begin{proof}
The proof is based on computing the expectations in \eqref{eq:SINRk} {relying on} the channel estimates and estimation errors in Lemma~\ref{lemma:ChannelEst}. The detailed proof is available in Appendix~\ref{appendix:SE}.
\end{proof}

The numerator of the SINR in~\eqref{eq:SINRmrc} represents the coherent power contributions {arriving} from both {the satellites} and terrestrial links, incorporating the effects of {both} LoS and NLoS components. The coherent combining gain scales as $\mathcal{O}((LM + NP)^2)$, while the aggregate multiuser interference in~\eqref{eq:MIk} grows with $\mathcal{O}(KL^2M + KLMP + NKP)$, capturing the influence of {both} inter-satellite and terrestrial co-channel interference. The additive noise scales linearly with $\mathcal{O}(LM + NP)$, proportional to the total number of receiving antennas. Consequently, user-centric association effectively mitigates interference and noise, enhancing the overall throughput.

\section{ENERGY {EFFICIENCY} OPTIMIZATION FOR SPACE-TERRESTRIAL COMMUNICATIONS}
This section formulates an optimization framework for cooperative space-terrestrial networks centralized processing, jointly addressing user association and uplink power control under realistic channel and system constraints.

\subsection{Total Energy Efficiency Optimization}\label{sec:problem}
We aim {for maximizing} the overall energy efficiency of the integrated network {considered}, defined as the ratio between the achievable uplink sum-rate and the {system's total} power consumption. The uplink sum-rate is given by $\sum_{k=1}^{K} R_k$ in~\eqref{eq:Rkv1}, while the total power consumption $\mathcal{P}^{\text{total}}$ accounts for all active network entities, including the users, APs, satellites, and their corresponding fronthaul and feeder links.

\subsubsection{Power Consumption Model}
The total power consumption of the system is expressed as the aggregate energy expenditure of all components~\cite{feng2025uplink}:
\begin{equation}\label{eq:p_total}
    \mathcal{P}^{\text{total}} = 
    \mathcal{P}^{\text{user}} +
    \mathcal{P}^{\mathrm{AP}, \text{c}} +
    \mathcal{P}^{\mathrm{SAT}, \text{c}} +
    \mathcal{P}^{\mathrm{AP}, \text{fh}} +
    \mathcal{P}^{\mathrm{SAT}, \text{fl}}.
\end{equation}
Each term in~\eqref{eq:p_total} corresponds to a specific network component and {it} is detailed as follows: The total user-side power consumption is modeled as the sum of the radiated transmit power and the fixed circuit power:
$\mathcal{P}^{\text{user}} =
    \sum\nolimits_{k=1}^{K}
    \left(
        {\rho_k}/{\gamma} +
        \mathcal{P}_{k}^{\mathrm{UE}, \text{tc}}
    \right),$    
where $\gamma$ is the power amplifier efficiency, and $\mathcal{P}_{k}^{\mathrm{UE}, \text{tc}}$ represents the static circuit power consumed by the RF transceiver and baseband processing. The circuit power consumption at all APs is modeled as: $\mathcal{P}^{\mathrm{AP}, \text{c}} = P
    \sum\nolimits_{n=1}^{N} 
    \mathcal{P}_{n}^{\mathrm{AP}, \text{tc}},$
where $\mathcal{P}_{n}^{\mathrm{AP}, \text{tc}}$ is the power required for analog and digital front-end operations such as converters, mixers, and filters. Similarly, the satellite-side circuit power is given by: $    \mathcal{P}^{\mathrm{SAT}, \text{c}} =
    M\sum\nolimits_{l=1}^{L} 
    \mathcal{P}_{l}^{\mathrm{SAT}, \text{tc}},$
where $\mathcal{P}_{l}^{\mathrm{SAT}, \text{tc}}$ denotes the static circuit power associated with each satellite antenna element. The power consumed by the fronthaul and feeder links connecting APs and satellites to the CPU can be expressed as
\begin{align}
    \label{eq:p_fh_ap}
    \mathcal{P}^{\mathrm{AP}, \text{fh}} &=
    \sum_{n=1}^{N}
    \left(
        \mathcal{P}_{n}^{\mathrm{AP}, 0} +
        \mathcal{P}_{n}^{\mathrm{AP}, \text{bt}}
        \sum_{k=1}^{K}
        \alpha_{nk}^{\mathrm{AP}}
        B \log_2(1+\mathrm{SINR}_k)
    \right), \\[-1.5ex]
    \label{eq:p_feeder_sat}
    \mathcal{P}^{\mathrm{SAT}, \text{fl}} &=
    \sum_{l=1}^{L}
    \left(
        \mathcal{P}_{l}^{\mathrm{SAT}, 0} +
        \mathcal{P}_{l}^{\mathrm{SAT}, \text{bt}}
        \sum_{k=1}^{K}
        \alpha_{lk}^{\mathrm{SAT}}
        B \log_2(1+\mathrm{SINR}_k)
    \right),
\end{align}
where $\mathcal{P}_{n}^{\mathrm{AP}, 0}$ and $\mathcal{P}_{l}^{\mathrm{SAT}, 0}$ denote the traffic-independent power coefficients, while $\mathcal{P}_{n}^{\mathrm{AP}, \text{bt}}$ and $\mathcal{P}_{l}^{\mathrm{SAT}, \text{bt}}$ represent the traffic-dependent power terms, both proportional to the total data rate~\cite{ngo2017total,zaeem2024energy}.

\subsubsection{Problem Formulation}
The total energy efficiency maximization problem is thus formulated as:
\begin{subequations}\label{problem-1}
    \begin{alignat}{2}
        & \underset{\mathbf{x}}{\text{maximize}} && f(\mathbf{x}) =
         \frac{\sum_{k=1}^{K} R_k(\mathbf{x})}{\mathcal{P}^{\text{total}}(\mathbf{x})}, \\
        & \text{subject to} \quad &&
        \alpha^\mathrm{AP}_{nk}, \alpha^\mathrm{SAT}_{lk} \in \{0,1\}, \quad \forall n,l,k, \\
        & && 0 \leq \rho_k \leq P_{\max,k}, \quad \forall k.
    \end{alignat}
\end{subequations}
Hence, $\mathbf{x}=\{\{\alpha_{lk}^{\mathrm{SAT}}\},\{\alpha_{nk}^{\mathrm{AP}}\},\{\rho_k\}\}$ represents a candidate solution encoding a complete network configuration, where $\alpha_{nk}^{\mathrm{AP}}$ and $\alpha_{lk}^{\mathrm{SAT}}$ denote the AP-user and satellite-user association indicators, respectively, while $P_{\max,k}$ is the maximum transmit power of user $k$. This mixed-integer nonlinear program is inherently nonconvex due to the coupling between binary association variables and continuous power control, making global optimization intractable for large systems and motivating low-complexity, near-optimal solutions.

\begin{remark}
In energy-efficiency optimization, the sum rate is often used for analytical tractability, but it can severely compromise fairness by favoring users with strong channels. To mitigate this issue, fairness-oriented utilities such as max-min fairness and geometric mean rate have been proposed to ensure more balanced performance under heterogeneous channel conditions~\cite{chen2024enhancing,zhu2023max,van2024performance}. More generally, a comprehensive multi-objective formulation that jointly accounts for conflicting metrics such as SINR, energy consumption, load balancing, latency, and fairness is left for future work~\cite{8957702}.
\end{remark}

\subsection{Low-Complexity Metaheuristic Algorithm}
\label{sec:alg}
We introduce an adaptation of the IDE algorithm to address problem~\eqref{problem-1}, where each individual in the population is represented by a combined binary-real encoding scheme.
\subsubsection{Individual {Representation} and Population {Initialization}}\label{sec:pop}
In the proposed IDE algorithm, the set of candidate solutions, denoted by $\mathcal{Q}$, comprises ${Q}$ individuals. Each individual $i \in \{1, \ldots, Q\}$ represents a feasible solution vector of length $K(L + N + 1)$. Formally, the $i$-th solution is expressed as:
\begin{equation}\label{eq:x_i}
    \mathbf{x}_{i} = \{x_{i,1}, x_{i,2}, \ldots, x_{i,K(L+N+1)}\},
\end{equation}
where $\mathbf{x}_{i}$ is partitioned into three segments corresponding to distinct variable groups involved in the optimization problem. In more detail, it can be shown that
\begin{itemize}
    \item The first segment of length $KL$ consists of binary elements $\{0,1\}$, represent the associations between user~$k$ and the $L$ satellites, i.e., $\alpha_{lk}^{\mathrm{SAT}} = x_{i,(k-1)L+l}$.
    \item The second segment of length $KN$ also consists of binary elements $\{0,1\}$, indicate the connections between user~$k$ and the $N$ APs, i.e., $\alpha_{nk}^{\mathrm{AP}} = x_{i,KL+(k-1)N+n}$.
    \item The final segment of length $K$ contains real values in $[0,1]$, representing the normalized transmit power of each user, given by $x_{i,K(L+N)+k} = \rho_k / P_{\max,k}$.
\end{itemize}
During initialization, ${Q}-1$ individuals are randomly generated with chromosome elements uniformly sampled from $[0,1]$, as the IDE algorithm operates on real-valued representations. Each individual’s fitness is evaluated using the objective function in~\eqref{problem-1}, where higher fitness values indicate stronger and more favorable solutions~\footnote{\textcolor{black}{During fitness evaluation, the continuous variables for association indicators are rounded to their binary values before computing the objective function.}}.

\subsubsection{Mutation operator}\label{sec:mutation}
The mutation operator produces a child individual $\mathbf{x}_c$ from one or more parents. To balance exploration and convergence, we adopt {a pair of} mutation strategies and {choose} between them with {certain} probabilities
\begin{itemize}
    \item The DE/$p$best/1 strategy, with probability $\lambda$:
    \begin{equation}\label{eq:pbest-1}
        \mathbf{x}_c = \mathbf{x}_{p\text{best}} + \textsf{F}(\mathbf{x}_{r_1} - \mathbf{x}_{r_2}).
    \end{equation}
    \item The DE/current-to-$p$best/1 strategy, with probability $1 - \lambda$:
    \begin{equation}\label{eq:current-to-pbest-1}
        \mathbf{x}_c = \mathbf{x}_p + \textsf{F}(\mathbf{x}_{p\text{best}} - \mathbf{x}_p) + \textsf{F}(\mathbf{x}_{r_1} - \mathbf{x}_{r_2}),
    \end{equation}
\end{itemize}
Here, $\mathbf{x}_p$ denotes the parent solution used as the mutation base, while $\mathbf{x}_{p\text{best}}$ is randomly drawn from the top $p\%$ of the population ranked by fitness. The scaling factor is denoted by $\textsf{F}$. The vectors $\mathbf{x}_{r_1}$ and $\mathbf{x}_{r_2}$ are randomly selected, with $\mathbf{x}_{r_1}$ taken from the current population and $\mathbf{x}_{r_2}$ from either the current or archived populations, where the archive stores previously successful parents to preserve genetic diversity.

At initialization, and after every $T$ generations, the probability $\lambda$ is adaptively updated according to the relative performance of the {pair of} mutation operators:
\begin{equation}\label{eq:lambda}
    \lambda =
    \begin{cases}
        0.8, & \text{if } {\Delta_1}/{\textsf{CE}_1} > {\Delta_2}/{\textsf{CE}_2},\\[3pt]
        0.2, & \text{otherwise,}
    \end{cases}
\end{equation}
where $\Delta_1$ and $\Delta_2$ denote the cumulative fitness gains achieved by the {DE/$p$best/1} and {DE/current-to-$p$best/1} strategies over the {most recent} $T$ generations, respectively. Similarly, $\textsf{CE}_1$ and $\textsf{CE}_2$ represent the number of times each operator has been invoked to produce mutant solutions. This adaptive mechanism introduces self-adjusting behavior into the mutation process, biasing the search toward the more effective operator.

\subsubsection{Crossover operator}\label{sec:crossover}
After mutation, an offspring vector $\mathbf{x}_o = \{x_{o1}, x_{o2}, \ldots, x_{oJ}\}$ is produced by combining the parent $\mathbf{x}_p$ and the mutant vector $\mathbf{x}_c$ through a binomial crossover, where $x_{oj},j=1,\cdots,J,$ are:
\begin{equation}\label{eq:crossover}
x_{oj} =
\begin{cases}
x_{cj}, & \text{if } \mathcal{U}(0,1) < \textsf{CR} \text{ or } j = j_{\text{rand}}, \\
x_{pj}, & \text{otherwise},
\end{cases}
\end{equation}
where $\mathcal{U}(\cdot)$ denotes a uniform random variable and $\textsf{CR}$ is the crossover rate. To maintain {population} diversity, {having} at least one randomly selected index $j_{\text{rand}}$ ensures $\mathbf{x}_o \neq \mathbf{x}_p$. Additionally, the offspring $\mathbf{x}_o$ replaces the parent $\mathbf{x}_p$ if it yields a higher fitness value.

\subsubsection{Parameter adaptation}
In standard DE, the scaling factor $\textsf{F}$ and crossover rate $\textsf{CR}$ are typically fixed. To enhance adaptability, we adopt the SHADE mechanism \cite{tanabe2013success}, which adaptively updates these parameters. In each generation, $\textsf{F}_i$ and $\textsf{CR}_i$ for the $i$-th mutation strategy are sampled from historical means $\mu_F$ and $\mu_{CR}$, updated via the archives $\mathcal{A}_F$ and $\mathcal{A}_{CR}$ as follows:
\begin{equation}\label{eq:gen-f-cr-1}
    \mathsf{F}_i  \sim \mathcal{C}(\mu_{F}, \upsilon),\text{ and }
    \mathsf{CR}_i \sim \mathcal{N}(\mu_{CR}, \upsilon),
\end{equation}
where $\mathcal{C}(\mu_F,\upsilon)$ and $\mathcal{N}(\mu_{CR},\upsilon)$ denote the Cauchy and normal distributions with means $\mu_F$ and $\mu_{CR}$ and scale $\upsilon$. The parameter $\mathsf{F}_i$ is resampled until nonnegative and clipped to $[0,1]$, while $\mathsf{CR}_i$ is also confined to $[0,1]$. After each generation, the $\mathsf{F}_i$ and $\mathsf{CR}_i$ values leading to fitness improvements are stored in $\mathcal{A}_F$ and $\mathcal{A}_{CR}$, respectively, and used to update $\mu_F$ and $\mu{CR}$ as follows:
\begin{align}
    \label{eq:update-f-cr-1}
    \mu_F = \frac{\sum_{\mathsf{F}_j \in \mathcal{A}_F} w_j \mathsf{F}_j^2}{\sum_{\mathsf{F}_j \in \mathcal{A}_F} w_j \mathsf{F}_j},\text{ and } 
    \mu_{CR} = \frac{\sum_{\mathsf{CR}_j \in \mathcal{A}_{CR}} w_j \mathsf{CR}_j^2}{\sum_{\mathsf{CR}_j \in \mathcal{A}_{CR}} w_j \mathsf{CR}_j} ,
\end{align}
where $\mathbf{w} = [w_1, \ldots, w_j]$ denotes the weights. For each successful trial, the offspring’s fitness gain guides $\mu_F$ and $\mu_{CR}$ toward more effective $\mathsf{F}$ and $\mathsf{CR}$ values, with updates applied per mutation operator.

\subsubsection{Convergence analysis}
We analyze the convergence of the proposed IDE algorithm toward an optimum within an $\varepsilon$-optimal solution space.

\begin{definition}
The space of the $\varepsilon$-optimal solutions $\mathcal{S}_{\varepsilon}^*$ is defined as: $\mathcal{S}_{\varepsilon}^* = 
    \{\, x \in \mathcal{S} \mid |f(x) - f(x^*)| \leq \varepsilon \,\},$
where $f(\cdot)$ denotes the objective function (e.g., energy efficiency), 
$\mathcal{S}$ is the feasible solution space, $x^*$ is the {globally} optimal solution, and $\varepsilon$ is a small positive tolerance. 
Thus, $\mathcal{S}_{\varepsilon}^*$ represents the set of all feasible solutions whose objective values differ from the optimal by no more than $\varepsilon$.
\end{definition}

\begin{theorem}
\label{theorem:convergence}
Given a population $\mathcal{Q}$ of $Q$ initial individuals, the probability that Algorithm~\ref{alg:ide} produces at least one $\varepsilon$-optimal solution is lower-bounded by:
\begin{equation}
    \label{eq:convg_lowerbound}
    \Pr(\mathcal{Q} \cap \mathcal{S}_\varepsilon^* \neq \emptyset)
    \ge 1 - (1 - \psi(\mathcal{S}_\varepsilon^*) P_{ep})^{Q},
\end{equation}
where $\psi(\mathcal{S}_\varepsilon^*)$ denotes the measure of the $\varepsilon$-optimal solution space $\mathcal{S}_\varepsilon^*$,  and $P_{ep} \in [0, 1]$ represents the mutation probability of each individual.
\end{theorem}

\begin{proof}
    For an individual $\mathbf{x}_p$, the mutation and crossover operations in Algorithm~\ref{alg:ide} generate a child $\mathbf{x}_c$ with the following probability density function: $pr(\mathbf{x}_c)=1$ if $\mathbf{x}_c \in \mathcal{S}$, and $0$ otherwise. The probability that the offspring $\mathbf{x}_c$ lies within the $\varepsilon$-optimal region is therefore:
       $ \Pr(\mathbf{x}_c \in S_\varepsilon^*) 
        = \int_{S_\varepsilon^*} pr(\mathbf{x}_c)\, d\mathbf{x}_c 
        = \psi(S_\varepsilon^*),$
    where $d\mathbf{x}_c$ represents the differential elements of an IDE chromosome, including connectivity indicators and power-allocation coefficients.
Therefore, after the mutation and crossover operations, the probability that the population does not contain any $\varepsilon$-optimal solution is given by: $\Pr(\mathcal{Q} \cap \mathcal{S}_\varepsilon^* = \emptyset)
        \le (1 - \Pr(\mathbf{x}_c \in \mathcal{S}_\varepsilon^*) P_{ep})^{Q} 
        = (1 - \psi(\mathcal{S}_\varepsilon^*) P_{ep})^{Q}.$
Consequently, the probability that at least one offspring lies in $S_\varepsilon^*$  satisfies ~\eqref{eq:convg_lowerbound}: $\Pr(\mathcal{Q} \cap \mathcal{S}_\varepsilon^* \neq \emptyset)
        = 1 - \Pr(\mathcal{Q} \cap \mathcal{S}_\varepsilon^* = \emptyset) 
        \ge 1 - (1 - \psi(\mathcal{S}_\varepsilon^*) P_{ep})^{Q}.$
\end{proof}
\begin{algorithm}[t]
\caption{Improved Differential Evolution Algorithm}\label{alg:ide}
\begin{algorithmic}[1]
    \Require Large-scale fading coefficient, bandwidth, number of subcarriers in an (OFDM) symbol, carrier frequency, objective function.

\Statex\myrule

    \Statex \texttt{// Initialization}
    \State Initialize $\mathcal{Q}^1$ as in [\ref{sec:pop}].
    \State Set parameters: $G,Q,\lambda, \mathcal{A}_{F} , \mathcal{A}_{CR}.$
    \State Set parameters: $\Delta_1, \Delta_2, \text{CE}_1, \text{CE}_2 \gets 0$.

    \Statex \texttt{// Main evolutionary loop}
    \For{$g \in \{ 1, \dots, G \}$}
        \State Initialize next generation $\mathcal{Q}'^g \gets \varnothing$.
        \For{$i \in \{ 1, \dots, Q \}$}
            \State Generate $\mathsf{F}_i$ and $\mathsf{CR}_i$ using \eqref{eq:gen-f-cr-1}
            \State Create mutant $\mathbf{x}_{ci}$ using [\ref{sec:mutation}].
            \State Create trial vector $\mathbf{x}_{oi}$ using [\ref{sec:crossover}].
            
            \If{$f(\mathbf{x}_{i}) < f(\mathbf{x}_{oi})$} 
            
            \Comment{Success: Offspring is better}
            
                \State $\mathcal{Q}'^g \gets \mathcal{Q}'^g \cup \{ \mathbf{x}_{oi} \}$.
                \State Update for operation $j$: $\Delta_j \gets \Delta_j + (f(\mathbf{x}_{oi}) - f(\mathbf{x}_i))$, $\text{CE}_j \gets \text{CE}_j + 1$.
            \Else \Comment{Failure: Parent survives}
                \State $\mathcal{Q}'^{g} \gets \mathcal{Q}'^{g} \cup \{ \mathbf{x}_{i} \}$
            \EndIf
        \EndFor

        \texttt{// Parameter adaptation}
        \State Update $\mathcal{A}_F, \mathcal{A}_{CR}$ using \eqref{eq:update-f-cr-1}

        \Statex \texttt{// Operator adaptation}
        \If{$\text{mod}(g, T) = 0$}
            \State Recalculate $\lambda$ by \eqref{eq:lambda}. Reset $\Delta_1, \Delta_2,\text{CE}_1, \text{CE}_2 \gets 0$
        \EndIf 

        \State $\mathcal{Q}^{g} \gets \mathcal{Q}'^{g}$ \Comment{Move to the next generation}
    \EndFor
    
    \Return Best sub-optimal solution $\mathbf{x}_{best}$
\end{algorithmic}
\end{algorithm}

\begin{corollary}\label{corollary:probability}
    The IDE algorithm converges in probability to the $\varepsilon$-optimal solution set $\mathcal{S}_\varepsilon^*$, that is:
    \begin{equation}
    \label{eq:convergence}
        \lim_{g \to \infty} 
        \Pr(\mathcal{Q}^g \cap \mathcal{S}_\varepsilon^* \neq \emptyset) = 1,
    \end{equation}
    where $\mathcal{Q}^g$ denotes the population at the $g$-th generation.
\end{corollary}

\begin{proof}
    As shown in Theorem~\ref{theorem:convergence}, for every generation $g = 1,2,\dots$, there exists a constant lower bound for the probability that at least one individual lies in the $\varepsilon$-optimal set:
    \begin{equation}
        \Pr(\mathcal{Q}^g \cap \mathcal{S}_\varepsilon^* \neq \emptyset)
        \ge 1 - (1 - \psi(\mathcal{S}_\varepsilon^\ast) P_{ep})^Q > 0.
    \end{equation}
    According to Theorem~2 and Corollary~3 in \cite{hu2013sufficient}, 
    if such a constant positive lower bound exists, the population sequence 
    $\{\mathcal{Q}^g\}$ converges in probability to $\mathcal{S}_\varepsilon^*$, 
    thus establishing  \eqref{eq:convergence}.
\end{proof}

\subsubsection{Complexity analysis}
The initialization requires $O(QK(L+N+1))$ operations. Each mutation-crossover step also costs $O(K(L+N+1))$, yielding a per-iteration complexity of $O(QK(L+N+1))$. Updating the parameter-adaptation history adds an additional $O(Q)$ overhead. Hence, the total computational complexity of the IDE algorithm is $O(GQK(L+N+1))$, where $G$ is the number of generations. The full procedure is summarized in Algorithm~\ref{alg:ide}.

\begin{remark}
    \textcolor{black}{From a theoretical perspective, DE is particularly well-suited to problem~\eqref{problem-1} for three reasons. First, the hybrid encoding in~\eqref{eq:x_i} as binary association variables and real-valued power levels, is naturally handled by DE's mutation/crossover operations followed by threshold-based binarization, without requiring problem-specific operators as in Genetic Algorithms (GAs). Second, via SHADE adaptation~\cite{tanabe2013success}, our IDE automatically adjusts \textsf{F} and \textsf{CR}, reducing the manual tuning burden typical of Particle Swarm Optimization (PSO) and GAs while mitigating sensitivity to initial parameter settings. Third, as established in Theorem~\ref{theorem:convergence} and Corollary~\ref{corollary:probability}, the IDE satisfies sufficient conditions for convergence in probability to the $\varepsilon$-optimal set, a theoretical guarantee absent from many competing metaheuristics.}
\end{remark}

\section{Numerical Results} \label{sec:result}
This section reports {our} experimental findings, {characterizing} the system’s performance and the {pair of} proposed algorithms against state-of-the-art benchmarks. The simulation parameters follow the 3GPP specifications in \cite{3gpp2025release20}, with large-scale and small-scale fading in the channel models. The satellite parameters are derived from practical simulations of the Starlink constellation \cite{al2021session}. We simulate networks with up to 3 satellites, 70 APs, and 50 users uniformly distributed over a 15 km$^2$ area, as shown in Fig.~\ref{fig:model}. The LEO satellites are located at $(x^{\mathrm{SAT}}, y^{\mathrm{SAT}}, 300)$ km, where $x^{\mathrm{SAT}}, y^{\mathrm{SAT}} \in [300,400]$, and use up to 100 antennas with 26.9 dBi gain, versus 10 dBi for ground devices~\cite{9473753}. The system operates over a 100 MHz bandwidth at 20 GHz with a coherence block of {10,000} OFDM subcarriers. Each symbol is transmitted at 20 dBW~\cite{bisognin2015millimeter}. 
The noise figures are 6 dB at the APs and 1.3 dB at the satellite. All parameters follow the 3GPP specifications in~\cite{3gpp2025release20} for rural propagation, including large-scale and small-scale fading.
For the energy-efficiency evaluation, the parameters are: users’ power amplifier efficiency $\gamma = 0.4$, transmit power $\rho_k = 0.2\,\text{W}$, and circuit power $\mathcal{P}_{\text{tc,k}}^{\text{UE}} = 0.1\,\text{W}$. The static circuit power per antenna is $\mathcal{P}^{\text{AP}}_{\text{tc,n}} = 0.2\,\text{W}$ at each AP and $\mathcal{P}^{\text{SAT}}_{\text{tc,l}} = 0.6\,\text{W}$ at each satellite. The fixed backhaul power of each AP is $\mathcal{P}^{\text{AP}}_{\text{0,n}} = 0.825\,\text{W}$, and the traffic-dependent coefficient is $\mathcal{P}^{\text{AP}}_{\text{bt,n}} = 0.25 \times 10^{-9}\,\text{W/bit/s}$. For satellites, the fixed feeder-link power and traffic-dependent coefficient are $\mathcal{P}^{\text{SAT}}_{\text{0,l}} = 3\,\text{W}$ and $\mathcal{P}^{\text{SAT}}_{\text{bt,l}} = 5 \times 10^{-9}\,\text{W/bit/s}$, respectively. 
The IDE algorithm's parameters are configured as follows. Thealgorithm uses a population size $Q = \min(4D, 300)$, where $D = K(L+N+1)$ is the chromosome length~[63], and a maximum number of generations isof $G = 500$ generations. Instead of fixed scaling factor $\textsf{F}$ and crossover rate $\textsf{CR}$, we use SHADE to adaptively sample $\textsf{F}_i \sim \mathcal{C}(\mu_{\textsf{F}}, \upsilon)$ and $\textsf{CR}_i \sim \mathcal{N}(\mu_{\textsf{CR}}, \upsilon)$ from historical successes, with initial values $\mu_{\textsf{F}} = 0.2$ and $\mu_{\textsf{CR}} = 1.0$~\cite{tanabe2013success}. The mutation-strategy probability $\lambda$ is initialized to $0.2$ and updated every $T=10$ generations according to~\eqref{eq:update-f-cr-1}.

\begin{figure}[!t]%
    \centering
    \includegraphics[width=0.5\textwidth]{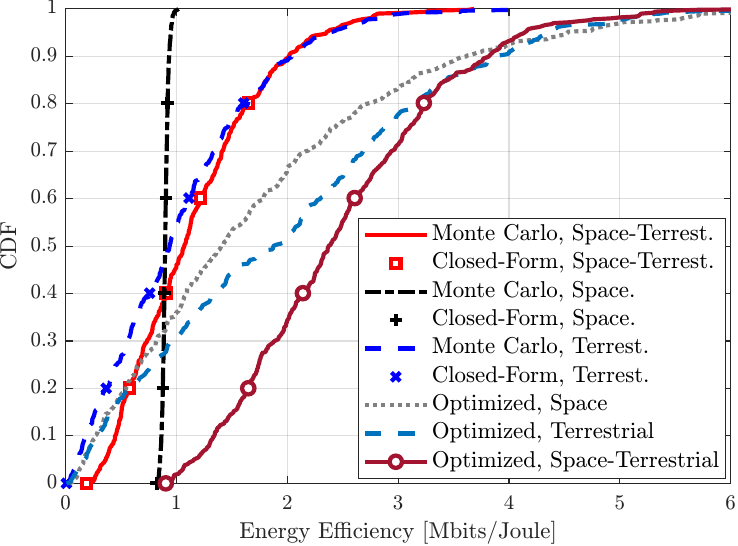}
    \caption{CDF comparison of Monte Carlo simulations, closed-form analysis, and optimized resource allocation under MRC.}
    \label{fig:MC-CF}
\end{figure}

\paragraph{\textbf{Analytical Derivations versus Monte Carlo Simulations}}

\begin{figure*}[ht]
    \centering
    \subfloat[]{\includegraphics[trim={0cm 0.5mm 0 0}, width=0.34\textwidth]{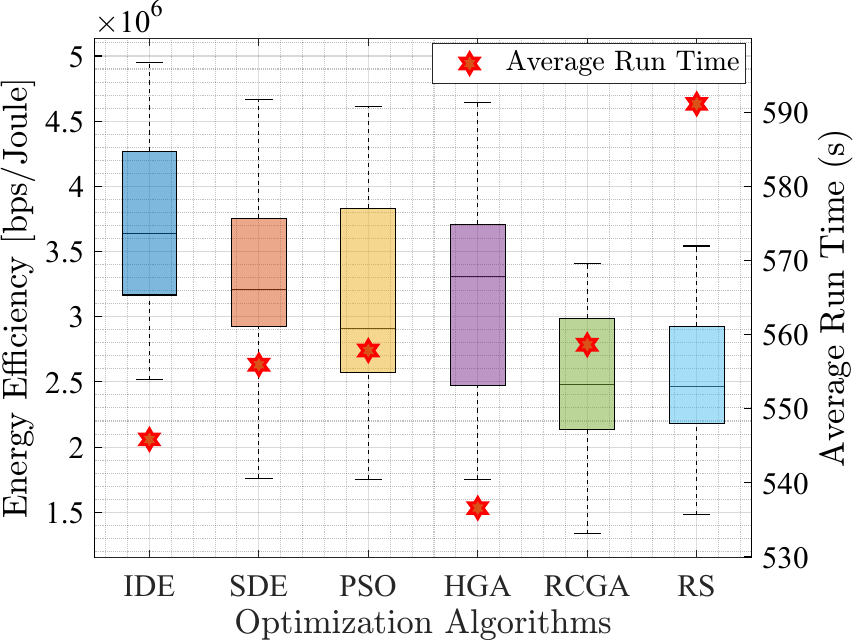}\label{fig:large_cp}}\hfill
    \subfloat[]{\includegraphics[width=0.317\textwidth]{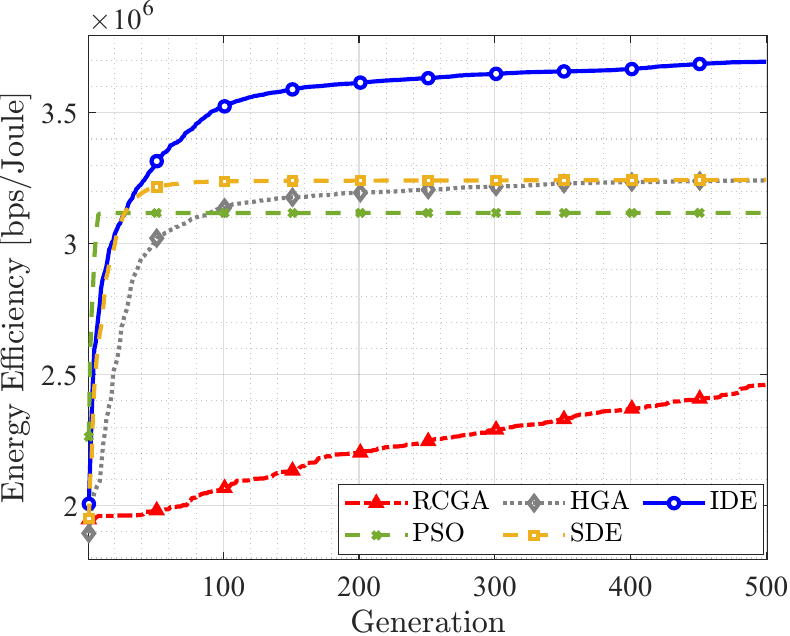}\label{fig:cvg_cp}}\hfill
    \subfloat[]{\includegraphics[width=0.32\textwidth]{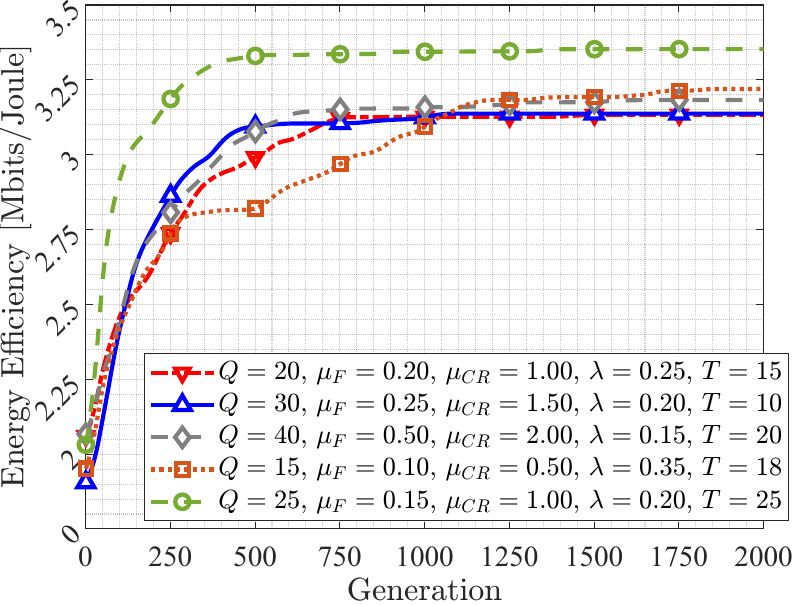}\label{fig:cvg_ide}}
    \caption{(a), (b) Comparison of the effectiveness proposed algorithm IDE via three metrics as optimal fitness value (Energy Efficiency), average run time and convergence with RS and other Meta Heuristic Algorithms. (c) Convergence behavior of the proposed IDE algorithm under different parameter
configurations.}
    \label{fig:compare_IDE_SearchAlg}
\end{figure*}

\begin{figure}[t]
    \centering
    \includegraphics[width=.9\linewidth]{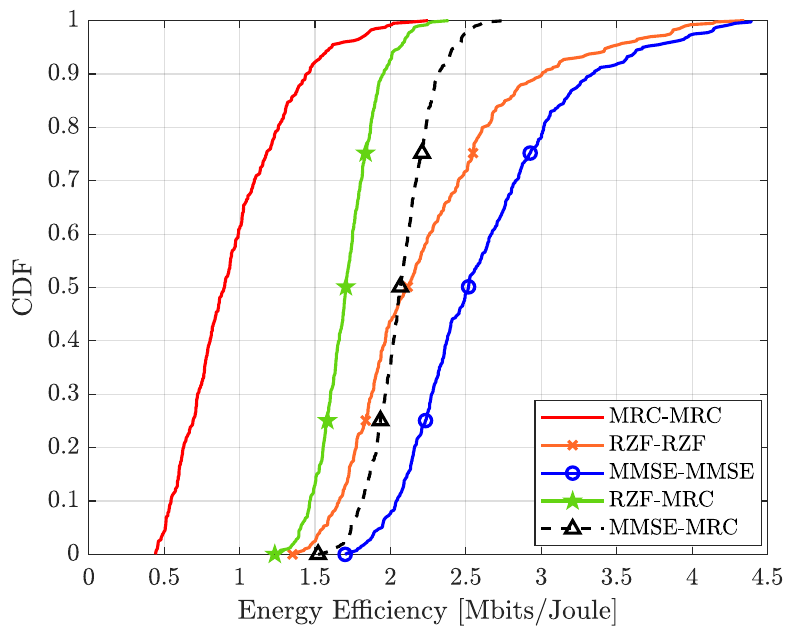}
    \caption{Performance comparison of different linear combining strategies.}
    \label{fig:CBV}
\end{figure}

Fig.~\ref{fig:MC-CF} illustrates the cumulative distribution function (CDF) of the energy efficiency for the space-terrestrial, terrestrial-only, and space-only systems, comparing Monte Carlo (MC) simulations with the corresponding closed-form analytical results.
It is observed that the analytical results closely match the MC simulations, confirming the accuracy of the closed-form expressions {derived}. The hybrid space-terrestrial architecture significantly outperforms both the terrestrial-only and space-only counterparts. Specifically, the proposed system achieves an average energy efficiency gain of $25\%$ and $20\%$ compared to the terrestrial-only and space-only baseline. \textcolor{black}{Furthermore, in practical deployments, satellites or terrestrial APs may become unavailable due to hardware failures, severe weather conditions, geographical obstacles, or natural disasters. In such cases, the proposed architecture can seamlessly switch to a space-only or terrestrial-only operating mode while preserving service continuity. The results demonstrate that user-centric optimization consistently achieves higher energy efficiency than conventional baseline architectures under these degraded conditions. Mathematically, this adaptability is enabled by the unified optimization framework, dynamically reconfiguring user association and resource allocation according to the set of active network nodes.}

\paragraph{\textbf{Comparison of IDE with Other Optimized Algorithms}}
We compare the performance of the proposed IDE algorithm against various benchmarks in Fig.~\ref{fig:compare_IDE_SearchAlg}
For the network configurations with 20 users, 40 APs, and two satellites, we benchmark the IDE algorithm against several meta-heuristic schemes (e.g., PSO \cite{shami2022particle}, Hybrid Genetic Algorithm (HGA)\cite{11134755}, Standard DE (SDE) \cite{mohamed2021differential} and Real-coded Genetic Algorithm (RCGA))\cite{van2024performance}, as well as Random Search (RS). As illustrated in Fig.~\ref{fig:large_cp}, the proposed IDE approach attains a better optimal solution, while significantly reducing the execution time. Moreover, Fig.~\ref{fig:cvg_cp} demonstrates that IDE exhibits a more stable and faster convergence {than} the other meta-heuristic algorithms across multiple iterations. To further assess parameter sensitivity, the proposed algorithm is evaluated under five distinct parameter configurations, as illustrated in Fig.~\ref{fig:cvg_ide}, to provide a more comprehensive examination of its robustness. The results demonstrate that the algorithm maintains consistent performance across all tested settings, with all configurations converging to stable energy efficiency values in the range of approximately $3.2-3.35$ Mbits/Joule, thereby confirming its robustness to parameter variation. While the convergence speed exhibits some dependence on the configuration chosen, most notably on the population size $Q$, the algorithm reliably reaches a stable point in all cases.

\paragraph{\textbf{Observing Various Linear Combining Techniques}}
We evaluate the system performance under various linear combining strategies applied at both the terrestrial and the space. Specifically, we compare the case where MRC is adopted for both the space and terrestrial segments with setups employing RZF or MMSE at both sides, as well as hybrid schemes, where {either a} RZF or MMSE {RC} is implemented at the CPU for terrestrial processing while MRC is used for the satellite signals. The results {of} Fig.~\ref{fig:CBV} demonstrate that changing the {RC} technique {is capable of enhancing} the system’s energy efficiency by up to a factor of 2.5 on average. These findings confirm the advantages of our linear {RC}, which effectively supports the cooperation between the satellite and distributed APs by leveraging the available channel estimates.

\paragraph{\textbf{User Performance Analysis and System Insights}}
{Table~\ref{tab:user_percent} presents the user association and energy efficiency results across three network scenarios, highlighting the interaction between the APs and the satellites in the hybrid network. The configurations {considered} are as follows: Scenario 1 with 10 users, 20 APs, and 1 satellite; Scenario 2 with 30 users, 40 APs, and 2 satellites; and Scenario 3 with 50 users, 60 APs, and 3 satellites.
It is observed that most users tend to associate either with {terrestrial} APs only or jointly with both APs and the satellite, while no users rely solely on the satellite link. This behavior underlines the dominant role of terrestrial APs in maintaining reliable connectivity and energy-efficient communication. Furthermore, as the network scales up from Scenario 1 to Scenario 3, the proportion of users connected to both APs and satellites increases significantly, reaching up to 70\%. This trend highlights the crucial role of joint AP-satellite cooperation in enhancing network coverage and overall system performance in large-scale configurations.}

Fig.~\ref{fig:SumconnectEachUser} shows per-user connectivity and performance under sum energy efficiency optimization in \eqref{problem-1}. Most users are mainly served by terrestrial APs, while only a few also use satellite links. Users connected to both layers achieve higher rates due to spatial diversity. However, maximizing sum energy efficiency can deactivate some users or assign negligible rates, since the optimization favors the most energy-efficient links.
To improve fairness, we also consider proportional energy efficiency (Fig.~\ref{fig:ProdconnectEachUser}), defined as $\prod_{k=1}^K R_k / \mathcal{P}^\mathrm{total}$. This yields a more balanced rate distribution but requires about 5.17 times more power than sum energy efficiency, which explains the prevalence of the latter in prior work. To move beyond fixed choices such as sum or proportional energy efficiency, we conduct a Pareto-based analysis (Fig.~\ref{fig:pareto}) to expose the trade-offs among the three metrics. The results show that higher minimum user rates generally require more transmit power, while power-minimizing solutions reduce both fairness and energy efficiency. The Pareto-optimal set enables flexible system design, and a deeper study of multi-objective and Pareto-based frameworks, as in~\cite{7570253,8957702,9762643,10258452}, is an important direction for future research.

To provide deeper insights into the performance contributions of individual system parameters, Fig.~\ref{EEPmax} investigates the energy efficiency behavior of the maximum transmit power ($P_{max}$) across four distinct operational scenarios. The results clearly demonstrate that jointly optimizing both variables consistently yields the highest performance. Specifically, this joint approach outperforms both the individual optimization schemes and the baseline fixed-architecture, delivering substantial energy efficiency improvements ranging from 20\% to 260\%. Furthermore, the extended simulation reveals a crucial underlying dynamic: while joint optimization represents the performance upper bound, power allocation plays a significantly more dominant role in the trade-off than user association. By proactively curtailing the transmit power, the system simultaneously addresses energy conservation and inter-user interference mitigation, thereby robustly safeguarding energy efficiency across the entire transmit power regime.
\begin{table}[t]
\caption{User association and energy efficiency for various scenarios.}\label{tab:user_percent}
\resizebox{\columnwidth}{!}{%
\begin{tabular}{lccc}
\hline \cline{1-4}
                                    \rule{0pt}{2.5ex}    & \multicolumn{3}{c}{\textbf{Scenario}} \\ \cmidrule(lr){2-4}
\multicolumn{1}{c}{\textbf{}}   \rule{0pt}{2.5ex}        & \textbf{{[}1{]}}     & \textbf{{[}2{]}}    & \textbf{{[}3{]}}    \\ \hline
\rule{0pt}{2.5ex} Users with AP only (\%)                & 65.00       & 48.33      & 30.00      \\
\rule{0pt}{2.5ex} Users with satellite only (\%)         & 0.00        & 0.00       & 0.00       \\
\rule{0pt}{2.5ex} Users with both (satellite+AP) (\%)    & 35.00       & 51.67      & 70.00      \\\hline
\rule{0pt}{2.5ex} Average Energy Efficiency {[}bps/J{]} & 3.56e6      & 2.84e6     & 2.04e6     \\ \hline \cline{1-4}
\end{tabular}%
}
\begin{tablenotes}
{\scriptsize \item \textcolor{black}{``Users with AP only" are connected to $\geq$ 1 terrestrial AP but no satellites; ``Users with satellite only" to $\geq$ 1 satellite but no APs; and ``Users with both (satellite+AP)" to both simultaneously. "\%" denotes each category's share of the total user population per scenario.}}
\end{tablenotes}
\end{table}

\begin{figure}[t]
    \centering
    \subfloat[]{\includegraphics[width=0.44\textwidth,trim={0cm 0 -1.1cm 0},clip]{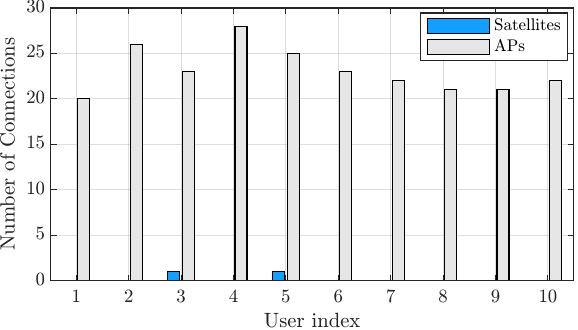}\label{fig:sumconnecta}} \\
    \subfloat[]{\includegraphics[width=0.43\textwidth,trim={0cm 0mm 1mm 0mm},clip]{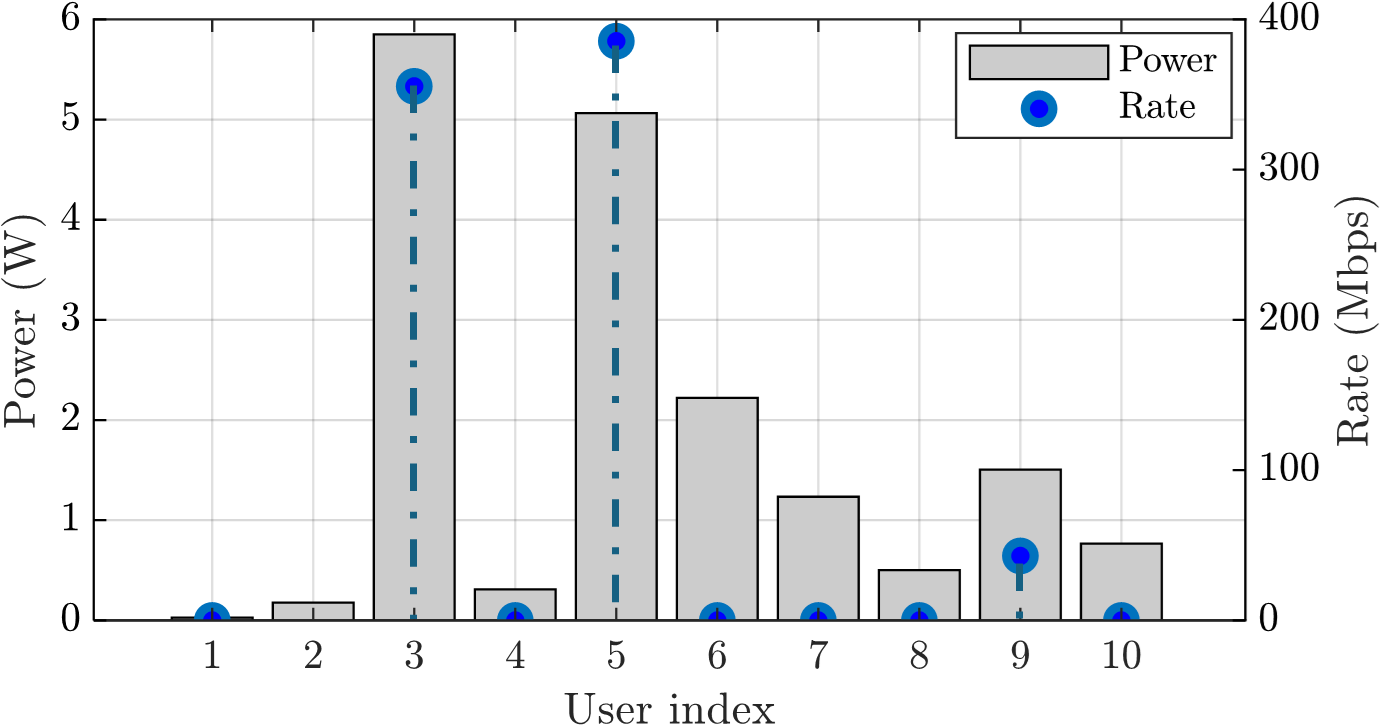}\label{fig:sumconnectb}}
    \caption{Per-user metrics when optimizing Sum Energy Efficiency. (a) Connection allocation. (b) Power and Rate.}
    \label{fig:SumconnectEachUser}
\end{figure}

\begin{figure}[t]
    \centering
    \subfloat[]{\includegraphics[width=0.44\textwidth,trim={0cm 0 -1.1cm 0},clip]{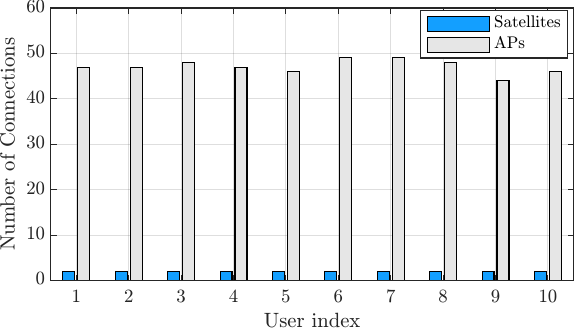}} \\
    \subfloat[]{\includegraphics[width=0.43\textwidth]{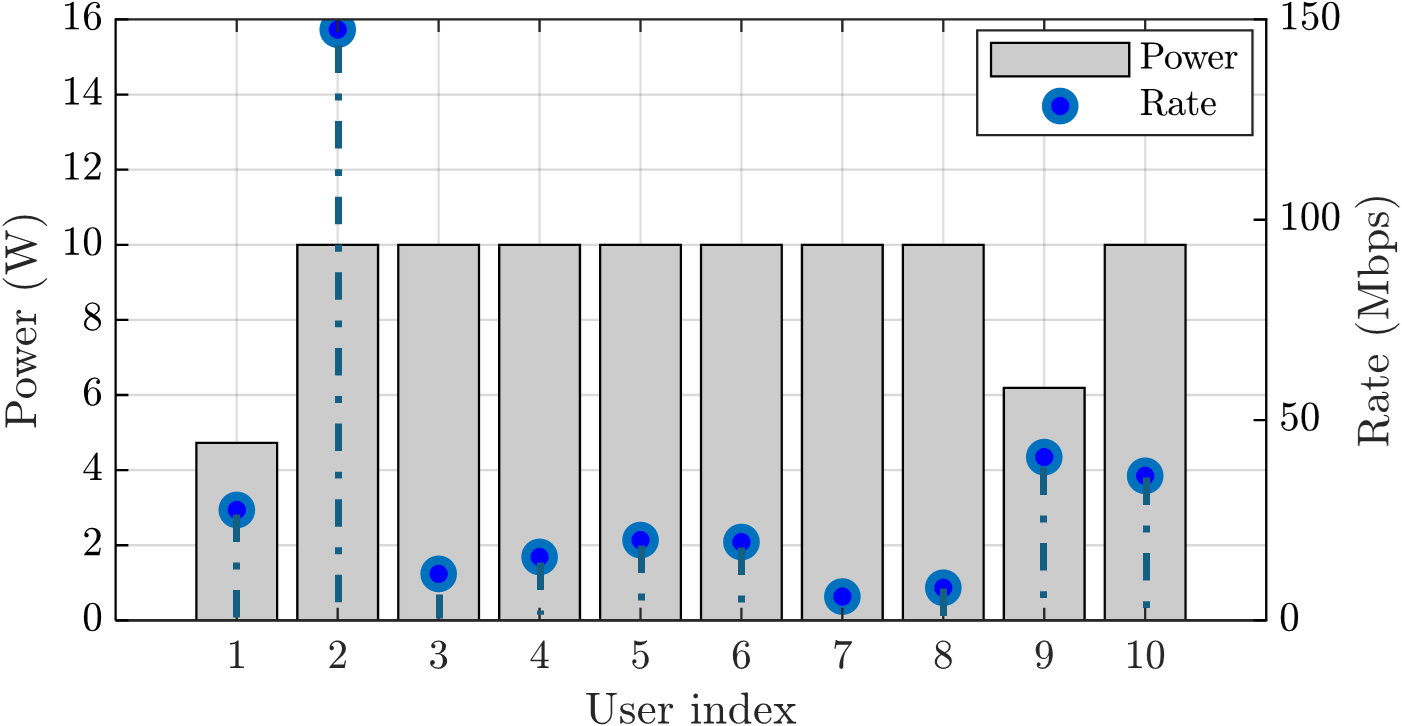}}
    \caption{Per-user metrics when optimizing Proportional Energy Efficiency. (a) Connection allocation. (b) Power and Rate.}
    \label{fig:ProdconnectEachUser}
\end{figure}

\begin{figure}
    \centering
    \includegraphics[width=\linewidth]{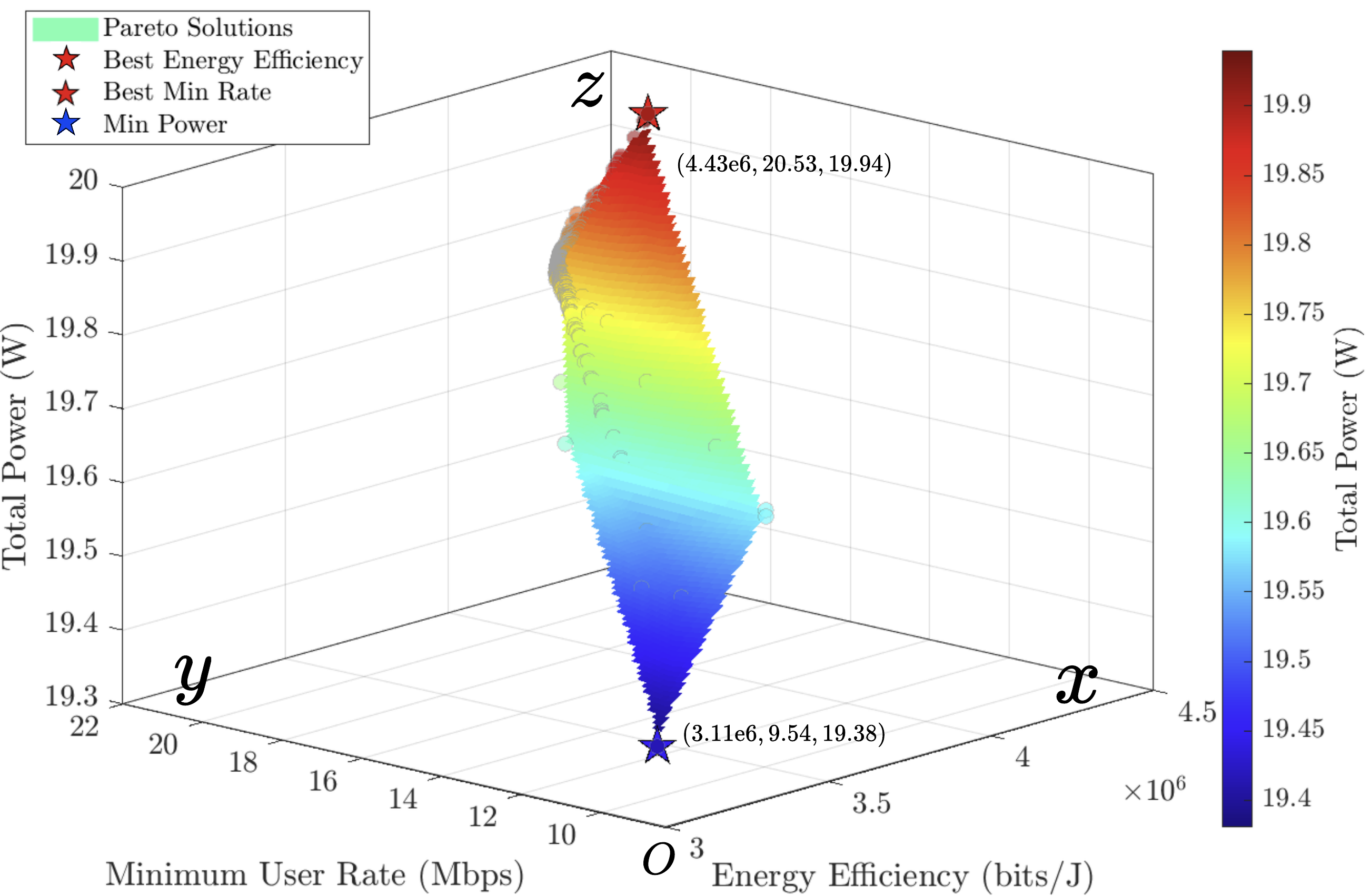}
    \caption{Pareto frontier of Energy Efficiency, Minimum rate, and Total power.}
    \label{fig:pareto}
\end{figure}

\begin{figure}
    \centering
    \includegraphics[width=.8\linewidth]{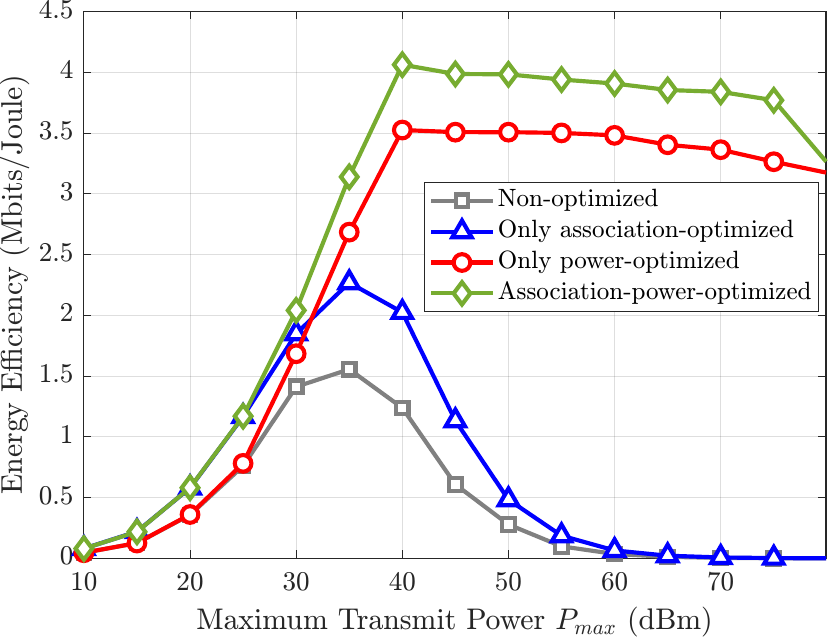}
    \caption{Energy efficiency versus maximum transmit power ($P_{max}$) under varying optimization strategies.}
    \label{EEPmax}
\end{figure}

\section{Conclusions}\label{sec:Conclusion}
A hybrid satellite-cell-free massive MIMO architecture was designed for NG networks. We derived a closed-form expression for the uplink ergodic throughput using MRC over spatially correlated Rician fading channels, and formulated a joint energy-efficient optimization problem for user association and power control. To solve the resultant NP-hard problem, an IDE algorithm with adaptive parameters and dual mutation strategies was proposed, achieving near-optimal performance {at} acceptable computational complexity and proven convergence. Simulation results confirmed that the proposed hybrid satellite-AP cooperation significantly enhances both ergodic throughput and energy efficiency, particularly for users {having} poor channel conditions, while the IDE algorithm outperforms benchmark metaheuristics in convergence speed and solution quality. This study opens promising directions for integrating intelligent reflecting surfaces and learning-assisted optimization to enable scalable and adaptive resource management in future space--terrestrial networks.

\appendix
\subsection[Proof of Lemma \ref*{lemma:ChannelEst}]{Proof of Lemma~$\ref{lemma:ChannelEst}$}\label{Apx:ChannelEst}
Firstly, we consider {terrestrial} channel estimation.
To derive the MMSE estimate of the ground channel $\mathbf{g}_{nk}$, we exploit the structure of the received pilot signal at the AP~$n$. The received signal $\mathbf{Y}_n^{\mathrm{AP},\text{pilot}}$ contains superimposed pilot transmissions from all users, which we decompose by projecting onto the specific pilot sequence $\pmb{\phi}_k$ of user $k$. Given the observation $\mathbf{y}_{nk}^{\mathrm{AP}} = \mathbf{Y}_n^{\mathrm{AP},\text{pilot}}\pmb{\phi}_k$ and the Gaussian prior distribution of $\mathbf{g}_{nk} \sim \mathcal{CN}(\mathbf{0}, \mathbf{G}_{nk})$, the MMSE estimator computes the conditional expectation
$\hat{\mathbf{g}}_{nk} = \mathbb{E}\{\mathbf{g}_{nk} | \mathbf{y}_{nk}^{\mathrm{AP}}\} = \mathbf{G}_{nk}(\sqrt{\rho K}\mathbf{I}_P)^H \big(\rho K\,\mathbf{G}_{nk}+\sigma_a^2\mathbf{I}_P\big)^{-1}\mathbf{y}_{nk}^{\mathrm{AP}}.$
Defining the auxiliary matrix $\mathbf{\Psi}_{nk} = (\rho K\,\mathbf{G}_{nk}+\sigma_a^2\mathbf{I}_P)^{-1}$, we obtain the compact representation
$\hat{\mathbf{g}}_{nk} = \sqrt{\rho K}\,\mathbf{G}_{nk}\,\mathbf{\Psi}_{nk}\,\mathbf{Y}_n^{\mathrm{AP},\text{pilot}}\pmb{\phi}_k.$
The posterior covariance of the estimate is given by
$\mathrm{Cov}(\hat{\mathbf{g}}_{nk}) = \rho K\,\mathbf{G}_{nk}\,\mathbf{\Psi}_{nk}\,\mathbf{G}_{nk} \triangleq \rho K\,\mathbf{\Omega}^{\mathrm{AP}}_{nk},$
where $\mathbf{\Omega}^{\mathrm{AP}}_{nk} = \mathbf{G}_{nk}\mathbf{\Psi}_{nk}\mathbf{G}_{nk}$. The estimation error $\mathbf{e}_{nk} = \mathbf{g}_{nk} - \hat{\mathbf{g}}_{nk}$ follows the distribution
$\mathbf{e}_{nk} \sim \mathcal{CN}\big(\mathbf{0}, \mathbf{C}^{\mathrm{AP}}_{e,nk}\big),$
where $\mathbf{C}^{\mathrm{AP}}_{e,nk} = \mathbf{G}_{nk} - \rho K\,\mathbf{G}_{nk}\mathbf{\Psi}_{nk}\mathbf{G}_{nk}$.
In this case, when $\mathbf{G}_{nk} = \beta^{\mathrm{AP}}_{nk}\mathbf{I}_P$ as spatially uncorrelated channels, the expressions simplify to:
\begin{align}
\mathbf{\Psi}_{nk} &= (\rho K\,\beta^{\mathrm{AP}}_{nk}+\sigma_a^2)^{-1}\mathbf{I}_P,\\
\hat{\mathbf{g}}_{nk} &= {\sqrt{\rho K}\,\beta^{\mathrm{AP}}_{nk}}{(\rho K\,\beta^{\mathrm{AP}}_{nk}+\sigma_a^2)^{-1}}\,\mathbf{Y}_n^{\mathrm{AP},\text{pilot}}\pmb{\phi}_k,
\end{align}
with $
\hat{\mathbf{g}}_{nk}\sim \mathcal{CN}\left(\mathbf{0}, {\rho K\,(\beta^{\mathrm{AP}}_{nk})^2}/(\rho K\,\beta^{\mathrm{AP}}_{nk}+\sigma_a^2)\,\mathbf{I}_P\right).$
Similarly, for the space link channel $\mathbf{h}_{lk}$ between user $k$ and the satellite~$l$, we apply the LMMSE estimator due to LoS component $\bar{\mathbf{h}}_{lk}$.
Given the observation $\mathbf{y}_{lk}^{\mathrm{SAT}} = \mathbf{Y}_l^{\mathrm{SAT},\text{pilot}}\pmb{\phi}_k$ and the prior $\mathbf{h}_{lk} \sim \mathcal{CN}(\bar{\mathbf{h}}_{lk}, \mathbf{H}_{lk})$, the LMMSE estimate is $\hat{\mathbf{h}}_{lk} = \bar{\mathbf{h}}_{lk}  + \mathbf{H}_{lk}(\sqrt{\rho K}\mathbf{I}_M)^H \big(\rho K\,\mathbf{H}_{lk}+\sigma_s^2\mathbf{I}_M\big)^{-1} \big(\mathbf{y}_{lk}^{\mathrm{SAT}}-\sqrt{\rho K}\,\bar{\mathbf{h}}_{lk}\big).$ Defining $\pmb{\Phi}_{lk} = (\rho K\,\mathbf{H}_{lk}+\sigma_s^2\mathbf{I}_M)^{-1}$, we obtain:
\begin{align}
\hat{\mathbf{h}}_{lk} = \bar{\mathbf{h}}_{lk} + \sqrt{\rho K}\,\mathbf{H}_{lk}\,\pmb{\Phi}_{lk}\,\big(\mathbf{Y}_l^{\mathrm{SAT},\text{pilot}}\pmb{\phi}_k - \sqrt{\rho K}\,\bar{\mathbf{h}}_{lk}\big).
\end{align}
The posterior covariance and estimation error follow analogously $\mathrm{Cov}(\hat{\mathbf{h}}_{lk}) = \rho K\,\mathbf{H}_{lk}\,\pmb{\Phi}_{lk}\,\mathbf{H}_{lk} \triangleq \rho K\,\mathbf{\Omega}_{lk},
\mathbf{e}_{lk} = \mathbf{h}_{lk} - \hat{\mathbf{h}}_{lk} \sim \mathcal{CN}\big(\mathbf{0}, \mathbf{C}^{\mathrm{SAT}}_{e,lk}\big),
$ where $\mathbf{C}^{\mathrm{SAT}}_{e,lk} = \mathbf{H}_{lk} - \rho K\,\mathbf{H}_{lk}\pmb{\Phi}_{lk}\mathbf{H}_{lk}$.

\subsection{Proof of Theorem~\ref{Theorem:SE}}\label{appendix:SE}
By considering the composite channel and substituting $k' = k$, the numerator in equation~\eqref{eq:SINRk} can be expressed as:
\begin{equation}\label{eq:e{okk}2}
|\mathbb{E}\{ o_{kk} \}|^2 = \sum\nolimits_{l=1}^L \Big (\| \bar{\mathbf{h}}_{lk} \|^2 + pK\,\mathrm{tr}(\boldsymbol{\Omega}_{lk}) \Big) + \sum\nolimits_{n=1}^{N} \mathrm{tr}(\mathbf{\Theta}_{nk}),
\end{equation}
which is obtained according to the channel distribution assumed and the MRC combining method. The first sum captures the coherent satellite signal power, combining the LoS component $\|\bar{\mathbf{h}}_{lk}\|^2$ and the NLoS estimation gain $\rho K\,\mathrm{tr}(\boldsymbol{\Omega}_{lk})$, while the second sum represents the terrestrial AP contribution through the MMSE estimate variance $\mathrm{tr}(\boldsymbol{\Theta}_{nk})$.
The first term in the denominator of~\eqref{eq:SINRk} is denoted as
\begin{equation}\label{eq:D1}
    \mathrm{D}_1 = \sum\nolimits_{k'=1}^{K} \rho_{k'}\,\mathbb{E}\{|o_{kk'}|^2\} = \rho_k\,\mathbb{E}\{|o_{kk}|^2\} + \mathrm{D}_2,
\end{equation}
where $\mathrm{D}_2 = \sum_{k'=1,\,k'\neq k}^{K} \rho_{k'}\,\mathbb{E}\{|o_{kk'}|^2\}$ collects the inter-user interference power from all $k' \neq k$.
To compute $\mathbb{E}\{|o_{kk}|^2\}$, we decompose $o_{kk}$ into four components and expand:
\begin{align}\label{eq:expandE{o^2}}
\mathbb{E}\{|o_{kk}|^2\}
= \mathbb{E}\{|a_{kk} + \tilde{a}_{kk} + b_{kk} + \tilde{b}_{kk}|^2\} = \mathbb{E}\{|a_{kk}|^2\} \notag\\+  \mathbb{E}\{|\tilde{a}_{kk}|^2\}  + \mathbb{E}\{|b_{kk}|^2\} + \mathbb{E}\{|\tilde{b}_{kk}|^2\}
+ 2\mathbb{E}\{a_{kk}b_{kk}\},
\end{align}
where $a_{kk} = \sum_{l=1}^L \|\hat{\mathbf{h}}_{lk}\|^2$ is the aggregate satellite LoS power,
$\tilde{a}_{kk} = \sum_{l=1}^L \hat{\mathbf{h}}_{lk}^H \mathbf{e}^{\mathrm{SAT}}_{lk}$ is the satellite estimation error term,
$b_{kk} = \sum_{n=1}^{N} \|\hat{\mathbf{g}}_{nk}\|^2$ is the terrestrial estimated channel power, and
$\tilde{b}_{kk} = \sum_{n=1}^{N} \hat{\mathbf{g}}_{nk}^H\mathbf{e}^{\mathrm{AP}}_{nk}$ is the AP-side estimation error term.
Cross-terms involving AWGN vanish due to this zero mean. The cross-terms between satellite and AP estimation errors also vanish by independence of the two links. We evaluate each expectation in turn.
The first expectation $\mathbb{E}\{|a_{kk}|^2\}$ captures the squared LoS satellite power, summed over all $L$ satellites. Expanding via the second moment of a sum of correlated Gaussian variables yields:
\begin{equation}\label{eq:1stterm}
\begin{aligned}
   \mathbb{E}\{|a_{kk}|^2\}
& = \Big(\sum\nolimits_{l=1}^L(\|\bar{\mathbf{h}}_{lk}\|^2 + 2pK\,\mathrm{tr}(\mathbf{\Omega}_{lk}))\Big)^2
\\& + \sum\nolimits_{l=1}^L\Big( 2pK\,\bar{\mathbf{h}}_{lk}^H \mathbf{\Theta}_k \bar{\mathbf{h}}_{lk}
+ (pK)^2\,\mathrm{tr}(\mathbf{\Omega}_{lk}^2)\Big). 
\end{aligned}
\end{equation}
The second expectation $\mathbb{E}\{|\tilde{a}_{kk}|^2\}$ represents the satellite-side NLoS estimation error variance. Using the error covariance $\mathbf{C}^{\mathrm{SAT}}_{e,lk} = \mathbf{H}_{lk} - \rho K\,\boldsymbol{\Omega}_{lk}$ from Lemma~\ref{lemma:ChannelEst}, it evaluates to:
\begin{equation}
\begin{aligned}\label{eq:a~kk}
    \mathbb{E}\{|\tilde{a}_{kk}|^2\}
= &\sum\nolimits_{l=1}^L\Big(\bar{\mathbf h}_{lk}^H\mathbf H_{lk}\bar{\mathbf h}_{lk}
-\rho K\,\bar{\mathbf h}_{lk}^H\boldsymbol{\Omega}_{lk}\bar{\mathbf h}_{lk} \\& +\rho K\,\operatorname{tr}(\mathbf H_{lk}\boldsymbol{\Omega}_{lk})
-(\rho K)^2\,\operatorname{tr}(\boldsymbol{\Omega}_{lk}^{2})\Big).
\end{aligned}
\end{equation}
The third expectation $\mathbb{E}\{|b_{kk}|^2\}$ corresponds to the terrestrial AP's estimated channel power. Applying standard second-moment identities for Gaussian vectors gives:
\begin{equation}
\mathbb{E}\{|b_{kk}|^2\}
=\Big(\sum\nolimits_{n=1}^N\,\operatorname{tr}(\boldsymbol{\Theta}_{nk})\Big)^2
+\,\sum\nolimits_{n=1}^N\,\operatorname{tr}(\boldsymbol{\Theta}_{nk}^2).
\end{equation}
The fourth expectation $\mathbb{E}\{|\tilde{b}_{kk}|^2\}$ captures the AP-side estimation error variance. Exploiting the independence between $\hat{\mathbf{g}}_{nk}$ and $\mathbf{e}^{\mathrm{AP}}_{nk}$, and using the error covariance $\mathbf{C}^{\mathrm{AP}}_{e,nk}$ from Lemma~\ref{lemma:ChannelEst}:
\begin{equation}\label{eq:b~kk}
\mathbb{E}\{|\tilde{b}_{kk}|^2\}
= \sum\nolimits_{n=1}^N \Big(\operatorname{tr}(\mathbf G_{nk}\boldsymbol{\Theta}_{nk})-\operatorname{tr}(\boldsymbol{\Theta}_{nk}^2)\Big).
\end{equation}
Finally, the cross-term $\mathbb{E}\{a_{kk}b_{kk}\}$ represents the coherent combining gain between the satellite LoS power and the terrestrial estimate power. Since $a_{kk}$ and $b_{kk}$ are independent across the two links, the expectation factorizes as:
\begin{equation}\label{eq:5thterm}
\mathbb{E}\{a_{kk}b_{kk}\}
= 2 \Big[\sum\nolimits_{l=1}^L \big(\|\bar{\mathbf h}_{lk}\|^2+\rho K\,\operatorname{tr}(\boldsymbol{\Omega}_{lk})\big)\Big]
\Big[\sum\nolimits_{n=1}^N \,\operatorname{tr}(\boldsymbol{\Theta}_{nk})\Big].
\end{equation}
Upon substituting \eqref{eq:1stterm}-\eqref{eq:5thterm} into \eqref{eq:expandE{o^2}} and collecting terms, the closed-form expression for $\mathbb{E}\{|o_{kk}|^2\}$ becomes 
\begin{equation}\label{eq:|o_kk|^2}
\begin{aligned}
\mathbb{E}\{|o_{kk}|^2\}  =& \Big(\sum\nolimits_{l=1}^L \big(\|\bar{\mathbf h}_{lk}\|^2+\rho K\,\operatorname{tr}(\boldsymbol{\Omega}_{lk})\big)+\sum\nolimits_{n=1}^N \,\operatorname{tr}(\boldsymbol{\Theta}_{nk})\Big)^2\\& +\sum\nolimits_{n=1}^N \operatorname{tr}(\mathbf G_{nk}\boldsymbol{\Theta}_{nk})+ \sum\nolimits_{l=1}^L  \Big(\bar{\mathbf h}_{lk}^H\mathbf H_{lk}\bar{\mathbf h}_{lk}
\\&+\rho K\,\bar{\mathbf h}_{lk}^H\boldsymbol{\Omega}_{lk}\bar{\mathbf h}_{lk} +\rho K\,\operatorname{tr}(\mathbf H_{lk}\boldsymbol{\Omega}_{lk})\Big).
\end{aligned}
\end{equation}
Here the first squared term is the coherent combining gain (satellite LoS/NLoS plus terrestrial), while the remaining terms account for the residual interference power owing to imperfect CSI at both the satellite and AP sides.
Next, we evaluate the mutual interference term $D_{2}$ in \eqref{eq:D1}. Since users transmit orthogonal pilot sequences, their channels are mutually independent, and the joint expectations factor accordingly. The five lines of \eqref{eq:D2} correspond, in order, to: (i) NLoS inter-user interference at the satellite; (ii) LoS–NLoS cross-interference projected onto $\boldsymbol{\Omega}_{lk}$; (iii) LoS-projected interference through $\mathbf{H}_{lk'}$; (iv) pure LoS inter-user interference; and (v) terrestrial AP inter-user interference via $\mathbf{G}_{nk'}$ as
\begin{align}\label{eq:D2}
&D_2 = \sum\nolimits_{\substack{k'=1,\\ k' \neq k}}^{K} \rho_{k'} \left( \mathbb{E}\{ |\hat{\mathbf{h}}_{lk}^H \mathbf{h}_{lk'}|^2 \}
+ \sum\nolimits_{m=1}^{M} \mathbb{E}\{ |\hat{\mathbf{g}}_{nk}^H \mathbf{g}_{nk'}|^2 \} \right) \nonumber \\
&= \underbrace{pK \sum_{\substack{k'=1,\\ k' \neq k}}^{K} \sum_{l=1}^L\rho_{k'}\, \mathrm{tr}(\mathbf{H}_{lk'} \boldsymbol{\Omega}_{lk'})}_{\text{(i) NLoS satellite interference}} \nonumber + \underbrace{pK \sum_{\substack{k'=1,\\ k' \neq k}}^{K} \sum_{l=1}^L \rho_{k'}\, \bar{\mathbf{h}}_{lk'}^H \boldsymbol{\Omega}_{lk} \bar{\mathbf{h}}_{lk'}}_{\text{(ii) LoS--NLoS cross-interference}} \nonumber \\
& + \underbrace{\sum\nolimits_{\substack{k'=1,\\ k' \neq k}}^{K} \sum\nolimits_{l=1}^L \rho_{k'} \, \bar{\mathbf{h}}_{lk}^H \mathbf{H}_{lk'} \bar{\mathbf{h}}_{lk}}_{\text{(iii) LoS-projected interference}} \nonumber + \underbrace{\sum\nolimits_{\substack{k'=1,\\ k' \neq k}}^{K} \rho_{k'} \, \Big(\sum\nolimits_{l=1}^L \bar{\mathbf{h}}_{lk}^H \bar{\mathbf{h}}_{lk'}\Big)^2}_{\text{(iv) pure LoS interference}} \nonumber \\
& + \underbrace{\sum\nolimits_{\substack{k'=1,\\ k' \neq k}}^{K} \sum\nolimits_{n=1}^{N} \rho_{k'}\, \mathrm{tr}(\mathbf{G}_{nk'}\mathbf{\Theta}_{nk})}_{\text{(v) terrestrial AP interference}}.
\end{align}
The noise contributions from the satellite and AP links are evaluated next. For the satellite, using the independence between $\hat{\mathbf{h}}_{lk}$ and the noise $\mathbf{w}^\mathrm{SAT}_{lk}$ (step ($a$)) yields:
\begin{align}\label{eq:hwsat}
\sum\nolimits_{l=1}^L \mathbb{E}\{ |\hat{\mathbf{h}}_{lk}^H \mathbf{w}^\mathrm{SAT}_{lk}|^2 \}
&\stackrel{(a)}{=} \sum\nolimits_{l=1}^L\mathbb{E}\{ \hat{\mathbf{h}}_{lk}^H \mathbb{E}\{ \mathbf{w}\mathbf{w}^H \} \hat{\mathbf{h}}_{lk} \} \nonumber \\
&= \sigma_s^2 \Big(\| \bar{\mathbf{h}}_{lk} \|^2 + \rho K\,\mathrm{tr}(\boldsymbol{\Omega}_{lk})\Big),
\end{align}
which scales with the total satellite receive power. Similarly, the noise power at the APs is:
\begin{align}\label{eq:gwAP}
\sum\nolimits_{n=1}^{N} \mathbb{E}\{ |\hat{\mathbf{g}}_{nk}^H \mathbf{w}^\mathrm{AP}_{nk}|^2 \}
= \sigma_a^2 \sum\nolimits_{n=1}^{N} \mathrm{tr}(\mathbf{\Theta}_{nk}),
\end{align}
which is proportional to the aggregate MMSE estimate variance across all APs.
By substituting \eqref{eq:|o_kk|^2} and \eqref{eq:D2} into \eqref{eq:D1}, we obtain the closed-form expression for the first denominator term of~\eqref{eq:SINRk}. Combining this with \eqref{eq:e{okk}2}, \eqref{eq:hwsat}, and \eqref{eq:gwAP} completes the proof.
\bibliographystyle{IEEEtran}
\bibliography{refs}
\end{document}